\documentclass{amsart}

  \usepackage{amsmath}
  \usepackage{amssymb}
  \usepackage[bbgreekl]{mathbbol}
  \DeclareSymbolFontAlphabet{\mathbbm}{bbold}
  \DeclareSymbolFontAlphabet{\mathbb}{AMSb}
  \usepackage{amsthm}
  \usepackage{mathrsfs}
  \usepackage{mathtools}
  \usepackage{stmaryrd}
  \usepackage{dsfont}
  \usepackage{amsmath}
  \usepackage{nicefrac}
  \usepackage[english]{babel}
  \usepackage{enumitem}
  \usepackage[a4paper,margin=3cm]{geometry}
  \usepackage{tikz}
  \usepackage{graphicx,float}
  \usetikzlibrary{matrix,arrows,decorations.pathmorphing}
  \usepackage{tikz-cd}
  \usepackage{caption, subcaption}
  \usepackage[vlined,boxed]{algorithm2e}

  \usepackage{lmodern}
  \usepackage{xcolor}
  \usepackage[plainpages=false,pdfpagelabels,colorlinks=true,citecolor=blue,urlcolor=magenta,hypertexnames=false]{hyperref}
  \usepackage{todonotes}
  \theoremstyle{plain}{
    \newtheorem{Theoreme}{Theorem}[section]
    \newtheorem{Lemme}[Theoreme]{Lemma}
    \newtheorem{Corollaire}[Theoreme]{Corollary}
    \newtheorem{Proposition}[Theoreme]{Proposition}
    \newtheorem{Definition-Proposition}[Theoreme]{Definition-Proposition}
  }
  \theoremstyle{definition}{
    \newtheorem{Definition}[Theoreme]{Definition}
    \newtheorem{Hypothesis}[Theoreme]{Hypothesis}
    \newtheorem{Notation}[Theoreme]{Notation}
  }
  \theoremstyle{remark}{
    \newtheorem{Remarque}[Theoreme]{Remark}
    
  }

    \newcommand{\Cbb}{\mathbb{C}}

    \newcommand{\Fbb}{\mathbb{F}}

    \newcommand{\Kbb}{\mathbb{K}}

    \newcommand{\Nbb}{\mathbb{N}}
    
    \newcommand{\Pbb}{\mathbb{P}}
    \newcommand{\Qbb}{\mathbb{Q}}
    \newcommand{\Rbb}{\mathbb{R}}
    \newcommand{\Sbb}{\mathbb{S}}

    \newcommand{\Zbb}{\mathbb{Z}}

    \newcommand{\Bcal}{\mathcal{B}}
    \newcommand{\Ccal}{\mathcal{C}}
    \newcommand{\Dcal}{\mathcal{D}}

    \newcommand{\Gcal}{\mathcal{G}}
    
    \newcommand{\Ical}{\mathcal{I}}

    \newcommand{\Lcal}{\mathcal{L}}
    
    \newcommand{\Ncal}{\mathcal{N}}
    \newcommand{\Ocal}{\mathcal{O}}

    \newcommand{\Tcal}{\mathcal{T}}

    \newcommand{\Gfrak}{\mathfrak{G}}

    \newcommand{\Pfrak}{\mathfrak{P}}

    \newcommand{\pfrak}{\mathfrak{p}}

\newcommand{\nocontentsline}[3]{}
\newcommand\stoptoc{%
   \let\origcontentsline\addcontentsline
   \let\addcontentsline\nocontentsline
}
\newcommand\resumetoc{%
   \let\addcontentsline\origcontentsline
}
  \DeclareMathOperator{\Disc}{Disc}
  \DeclareMathOperator{\Supp}{Supp}
  \DeclareMathOperator{\Tr}{Tr}
  \newcommand{\Id}{\mathrm{Id}}
  \newcommand{\Gal}[2]{\mathrm{Gal}(\nicefrac{#1}{#2})}
  \newcommand{\ud}{\,\mathrm{d}}
  \renewcommand{\Im}{\mathrm{Im}\,}
  
  \renewcommand{\leq}{\leqslant}
  \renewcommand{\geq}{\geqslant}
  
  \newcommand{\ldbrack}{\llbracket}
  \newcommand{\rdbrack}{\rrbracket}
  \newcommand{\overbar}[1]{\mkern 1.5mu\overline{\mkern-1.5mu#1\mkern-1.5mu}\mkern 1.5mu}

\newcommand{\gcrd}{\mathrm{gcrd}}
\newcommand{\DyN}[2]{\nicefrac{#1\langle\partial\rangle #2}{(\partial^p-y_N)}}

\newcommand{\Diff}{\mathrm{Diff}}
\title[Solving the $p$-Riccati equation]
{Factoring linear differential operators in positive characteristic
by means of solving a norm equation.}
\author{Raphaël Pagès}
\begin{document}
\maketitle
\begin{abstract}
  The solutions of the equation $f^{(p-1)}+f^p=h^p$ in the unknown function
  $f$ over an algebraic function field of characteristic $p$ are very
  closely linked to the structure and factorisations of linear differential
  operators with coefficients in function fields of characteristic $p$.
  However, while being able to solve this equation over general algebraic
  function fields is necessary even for operators with rational
  coefficients, no general resolution method has been developed. We present
  an algorithm for testing the existence of solutions in polynomial time in the
  ``size'' of~$h$ and a Las Vegas algorithm based on the computation of Riemann-Roch
  spaces and the selection of elements in the divisor class group, for
  computing solutions in time polynomial in the size of $h$ and linear in the
  characteristic~$p$. The size of the solutions yielded by this algorithm
  is polynomial in the ``size'' of~$h$. We discuss the applications of
  those algorithms to the factorisation of
  linear differential operators in positive characteristic $p$.
\end{abstract}
\tableofcontents
\section{Introduction}
This article deals with some algorithmic questions related to the
factorisation of linear differential operators in positive characteristic
$p$.
Let $K$ be a field equipped with an additive map $a\mapsto a'$ verifying the
Leibniz rule $(ab)'=a'b+ab'$. Such a map is called a derivation on $K$ and
$K$ is called a differential field. An example of such a field is $k(x)$,
where $k$ is any field, equipped with the derivation $\frac{\ud}{\ud x}$. 
We can consider the ring $K\langle\partial\rangle$ of linear differential 
operators with coefficients in~$K$, whose elements are polynomials in the
variable $\partial$ of the form
\[a_n\partial^n+a_{n-1}\partial^{n-1}+\cdots+a_1\partial+a_0\] with 
$a_i\in K$, and where
the (noncommutative) multiplication verifies the commutation rule
$\partial a=a\partial+a'$ for any $a\in K$. This formalism allows for a
wider variety of differential problems to be considered than the usual
analytical one, among which differential equations with coefficients in a
field of positive characteristic or with $p$-adic coefficients. The latter
has found many applications, $\emph{e.g.}$ to count points on elliptic
curves~\cite{Lauder04zeta}, to compute isogenies~\cite{LaVa16-old,Eid21}
and, more generally, to study (the cohomology of) many algebraic
varieties.\\
The problem of the factorisation of linear differential operators with coefficients
in $K=\Cbb(x)$ or $K=\overbar{\Qbb}(x)$ has been
well studied and several algorithms have been proposed over the
years~\cite{ChGoMe22,Gri90,VHo97}. The question of factorisation for
operators with coefficients in $\Fbb_p(x)$ has also been studied in the
perspective of developing modular algorithms to factor operators in
$\Qbb(x)\langle\partial\rangle$~\cite{Put96,Cluzeau03} after van der Put
published in~\cite{Put95} a full classification of finite dimensional differential modules in
characteristic~$p$ which serves as the basis of all factorisation
algorithms for operators in $\Fbb_p(x)\langle\partial\rangle$. The most
notable difference between operators in characteristic $0$ and in
characteristic $p$ is the size of the field of constants. Indeed, whereas it
is reduced to $\Cbb$ over $\Cbb(x)$, the field of constants of $\Fbb_p(x)$
is $\Fbb_p(x^p)$ over which the field of rational functions $\Fbb_p(x)$ is of finite
dimension $p$.  As a consequence, any operator
$L\in\Fbb_p(x)\langle\partial\rangle$ is a divisor
of an element $N\in\Fbb_p(x^p)[\partial^p]$, the center of $\Fbb_p(x)\langle\partial\rangle$.
Factoring those central elements (as a product of central elements) is much easier as they behave as bivariate
polynomials. Furthermore, the factorisation of central multiples of $L$ can
be used to recover information on the factorisations of $L$. This allows
to reduce the problem of factorisation in
$\Fbb_p(x)\langle\partial\rangle$ to the factorisation of divisors of some
$N(\partial^p)$ where $N$ is an irreducible polynomial over $\Fbb_p(x^p)$. 
To further improve the factorisation, we take advantage of the fact that
$\Fbb_p(x)\langle\partial\rangle$ is an Azumaya algebra (which is to say
here that all of the quotient rings
$\nicefrac{\Fbb_p(x)\langle\partial\rangle}{N(\partial^p)}$ are central
simple algebras over $\nicefrac{\Fbb_p(x^p)[\partial^p]}{N(\partial^p)}$, where $N$ is an
irreducible polynomial over $\Fbb_p(x^p)$) and comes equipped
with a reduced norm (which is locally equal to the reduced norms of the
local central simple algebras). A central element of the form
$N(\partial^p)$ where $N$ is an irreducible polynomial over $\Fbb_p(x^p)$
is then reducible as a differential operator if and only if it is a reduced
norm. Furthermore, in this case, its factors (and thus of any of its right
factors) are
all solutions of the induced norm equation.\\
In the case where $N(\partial^p)$ is of the form $\partial^p-a$ with $a\in
\Fbb_p(x^p)$, the solutions of the norm equations, if they exist, are
operators of order $1$. It was shown~\cite{Put95,Cluzeau03}, that the norm
equation can then be written as such
\[\Ncal(\partial-f)=\partial^p-\frac{\ud^{p-1}}{\ud
x^{p-1}}f-f^p=\partial^p-a\]
with $f\in\Fbb_p(x)$. This result generalizes to higher degrees of $N$. For the
sake of simplicity we shall assume $N$ to be separable. Let $y_N$ be a root
of $N$ in a separable closure of $\Fbb_p(x^p)$. The solutions of the norm
equation in $\Fbb_p(x)\langle\partial\rangle$ are in bijection with the
solutions of
\begin{equation}\label{p-riccati_equation}
  \frac{\ud^{p-1}}{\ud x^{p-1}} f+f^p=y_N
\end{equation}
in $\Fbb_p(x)[y_N]$. We call this equation the $p$-Riccati equation relative to $N$.
Furthermore, if $L$ is a divisor of $N(\partial^p)$ (not necessarily
irreducible) then the solutions of the $p$-Riccati equation can be used to
recover irreducible divisors of $L$. One way of doing that is to notice
that if $f$ is a solution of the $p$-Riccati equation relative to $N$ then
$L(\partial-f)$ has an algebraic solution $b\in\Fbb_p(x)[y_N]$. It follows
that an irreducible divisor of $L$ is given by the smallest left multiple
of $\partial-f-\frac{b'}{b}$ in $\Fbb_p(x)\langle\partial\rangle$. The
``size'' of the irreducible divisor of $L$ that this method returns thus depends at
least in part on the ``size'' of the solution to the $p$-Riccati equation
used. In particular, while it is not sufficient it is important in the
perspective of developing modular methods for factorisation that the ``size''
of the solution to the $p$-Riccati equation is independent from $p$.\\
Finally the existence of solutions to the $p$-Riccati equation acts as an
irreducibility test for $N(\partial^p)$.
\subsection{State of the art}
  After publishing a full classification of finite dimensional differential
  modules, Marius van der Put worked on the algorithmic aspects of the
  factorisation in positive characteristic, with the perspective of
  developing modular algorithms for the characteristic $0$. He wrote in
  1997 a manuscrit~\cite{Put97} which would go unpublished, in which he describes a nearly complete
  algorithm to factor linear differential operators in positive
  characteristic. This algorithm makes use of the $p$-curvature of an
  operator to compute a central multiple of it and do the aforementioned
  reduction to the case of a divisor of some $N(\partial^p)^m$ where $N$ is
  an irreducible polynomial over $\Fbb_p(x^p)$. This part of the algorithm
  is called the computation of the \emph{isotypical decomposition} of the
  operator.\\
  In this same manuscript, van der Put wrote the only method known up to
  this point to solve the $p$-Riccati equation relative to a general $N$,
  and suppose known a non trivial divisor $L$ of $N(\partial^p)$. In this
  case, if $y_N$ is a root of $N$, then one may consider the operator
  $L_*:=\gcd(L,\partial^p-y_N)\in\Fbb_p(x)[y_N]\langle\partial\rangle$. By writing
  \[L_*=\partial^m+b_{m-1}\partial^{m-1}+\cdots+b_1\partial+b_0\]
  it can be shown that $-\frac{b_{m-1}}{m}$ is a solution to the
  $p$-Riccati equation relative to $N$. As previously stated, this method can only be used if one already 
  knows a nontrivial divisor of
  $N(\partial^p)$. In particular if nothing else is known it cannot be used
  as an irreducibility test for $N(\partial^p)$. Furthermore, computing the
  greatest common divisor of $L$ with an operator of order $p$ yields an
  operator whose coefficients are of linear ``size'' in $p$. Thus the
  solution to the $p$-Riccati equation that this method returns has
  linear size in $p$ as well. This size dependency on $p$ spreads to the
  corresponding factors of $N(\partial^p)$.\\

  In the case where $N$ is a polynomial of degree $1$, an algorithm to solve the $p$-Riccati equation was
  proposed~\cite[§13.2.1]{PuSi03} which does not have those drawbacks. In that setting, the
  $p$-Riccati equation can be written as 
  \[\frac{\ud^{p-1}}{\ud x^{p-1}} f+f^p=g^p\] with $g\in\Fbb_p(x)$. This method consists in
  showing that if rational solutions exists then one of them has the same
  denominator as $g$ and a numerator of degree at most the maximum of the
  degrees of the numerator and the denominator of $g$. Finding this solution is now an easy
  task since the map $f\mapsto\frac{\ud^{p-1}}{\ud x^{p-1}} f+f^p$ is
  $\Fbb_p$-linear. This method returns a solution of degree polynomial (in
  fact linear here) in that of $g$ and a naive computational approach
  outputs the result in polynomial time in the degree of
  $g$ and linear time in $p$.

  Since then, works around the factorisation of differential operators in
  positive characteristic have mostly shied away from the prospect of
  solving the $p$-Riccati equation in order to finish the factorisation.
  In~\cite{Cluzeau03}, Thomas Cluzeau adapted van der Put's {isotypical
  decomposition} algorithm to
  the case of differential systems of the form $Y'=AY$ where $A\in
  M_n(\Fbb_q(x)$, with $q$ being a power of $p$. To further the
  factorisation, Cluzeau suggested applying a similar process to other
  endomorphism of the system $Y'=AY$ (the set of which constitutes a ring
  called the \emph{Eigenring}. While this method usually works well in
  characteristic $0$, later experiences have shown that this method may not
  be as efficient in positive characteristic due to the difference in
  nature of the constant field ($\Cbb$ for systems with coefficients in
  $\Cbb(x)$ which is algebraically closed, and
  $\Fbb_q(x^p)$ for systems with coefficients in $\Fbb_q(x)$ which is not).
   
  In 2003~\cite{GiZh03}, Mark Giesbrecht and Yang Zhang studied the factorisation of Ore
  polynomials (of which differential operators are a specific case) in
  positive characteristic. In their paper, the authors establish a direct connection between
  nontrivial factors of a given Ore polynomial and nontrivial zero divisors
  of the Eigenring. However, later work~\cite{GoTLoNa15,GoTLoNa19} showed
  that computing zero divisors was not always easy, for example  when the Eigenring 
  is a simple Artinian ring, which is precisely the case
  when factoring central operators.

  The case of the factorisation of central operators of the form
  $N(\partial^p)$ has been ignored by the previous works on factorisation.
  However, this case is highly nontrivial as we will see, and was until now the last
  missing piece for a complete factorization algorithm of linear
  differential operators in positive characteristic.

  \subsection{Contribution}
  We present two new algorithms regarding the factorisation of central
  operators in characteristic $p$ that are irreducible as polynomials, that is to say operators of the form
  $N(\partial^p)$ where $N$ is an irreducible polynomials with coefficients
  in $\Fbb_q(x^p)$ (with $q$ being a power of $p$). The first is a polynomial time
  irreducibility test.
  \begin{Theoreme}
    Let $q\in\Nbb^*$ be a power of $p$ and $N_*\in\Fbb_q[x,Y]$ be an irreducible
    bivariate polynomial of degree $d_x$ with respect to $x$ and $d_y$ with
    respect to $Y$. There exists an algorithm testing the irreducibility of
    $N_*^p(\partial)$ as a linear differential operator in polynomial time in $d_x, d_y$ and $\log(q)$.
  \end{Theoreme}
  We then use this irreducibility test to design an algorithm computing an
  irreducible factor of $N(\partial^p)$ when $N(\partial^p)$ is reducible
  as a linear differential operator. This algorithm works by computing a
  solution to
  the $p$-Riccati equation relative to $N$. We will also discuss the
  implications of those algorithms to the general factorisation of
  nontrivial factors of $N(\partial^p)$.
  \begin{Theoreme}
    Let $q\in\Nbb^*$ be a power of $p$ and let $N_*\in\Fbb_q[x,Y]$ be an
    irreducible polynomial of degree $d_x$ with respect to $x$ and $d_y$
    with respect to $Y$. We denote by $N\in\Fbb_q[x^p,Y]$ the unique
    polynomial such that $N_*^p(Y)=N(Y^p)$.
    \begin{itemize}
      \item There exists a solution to the $p$-Riccati equation relative to
        $N$ of size polynomial in $d_x$ and~$d_y$ and a Las Vegas algorithm taking
        $N_*$ as input and outputting this solution in linear time in~$p$
        and polynomial time in $d_x$ and $d_y$.
      \item $N(\partial^p)$ has irreducible factors in
        $\Fbb_q(x)\langle\partial\rangle$ of size polynomial in $d_x$ and
        $d_y$. There exists a Las Vegas algorithm taking
        $N_*$ as input and outputting such a factor in linear time in $p$
        and polynomial time in $d_x$ and $d_y$.
    \end{itemize}
  \end{Theoreme}
  \begin{Remarque}
    It should be noted that while we limit, for the sake of simplicity, our
    complexity study to the case of operators whose coefficients are
    rational functions over $\Pbb^1$,
    all of the aforementioned algorithms can in fact be designed for factoring operators
    whose coefficients are rational functions over an algebraic
    curve $\Ccal$.
  \end{Remarque}
\subsubsection*{Complexity basics}

We use the soft-$O$ notation $\tilde{O}$ which indicates that
polylogarithmic factors are not displayed. More precisely, if
$\lambda,\mu:\Nbb\rightarrow \Rbb_+$ are increasing functions, saying that
$\lambda(n)=\tilde{O}(\mu(n))$ means that there exists an integer
$k\in\Nbb$ such that $\lambda(n)=O(\mu(n)\log^k(\mu(n)))$.\\
We denote by $2\leq\omega\leq 3$ a feasible exponent for matrix
multiplication, that is, by definition, a real number for which
we are given an algorithm that computes the product of two $m$-by-$m$ matrices
over a ring $R$ for a cost of $O(m^\omega)$ operations in $R$.
From
\cite{VVYZZ23}, we know that we can take $\omega<2.371552$.
We shall also need estimates on the cost of computing characteristic
polynomials. Let denote $\Omega\in\mathbb{R}^*_+$ such that the computation of the
characteristic polynomial of a square matrix of size~$m$ with coefficients in
a ring $R$ can be done in $\tilde{O}(m^{\Omega})$ arithmetic operations in
$R$.
From \cite[Section~6]{KaVi04}, we know that it is theoretically possible to
take $\Omega\simeq 2.697263$.
 Finally, we assume that any two polynomials of degree~$d$ over a ring~$R$
  (resp. integers of bit size $n$) can be multiplied in $\tilde{O}(d)$
operations in $R$ (resp. $\tilde{O}(n)$ bit operations); FFT-like algorithms
allow for these complexities~\cite{CaKa91,HaHo21}.
\stoptoc
  \subsection*{Acknowledgments:} I am indebted to Alin Bostan and Xavier Caruso
  for their continuous support during my PhD. I also thank Martin Weimann for
  his input and his help regarding local computation in
  section~\ref{section_irreducibility_test}, in particular suggesting computing
  rational Puiseux expansions. I was supported by the ANR
  \emph{DeRerumNatura}~-~ANR-19-CE40-0018 and the
  ANR~\emph{ClapClap}~-~ANR-18-CE40-0026 as well as the Austrian FWF grants
  10.55776/PAT9952223.\\

  We begin by recalling some facts about differential operators in
  characteristic $p$.\\
\resumetoc
  \section{Reduction to a norm equation}
  In this section we will work on differential operators with coefficients
  in a differential field $(K,\partial)$ of characteristic $p$ verifying the following
  hypothesis:
  \begin{Hypothesis}\label{hypothesis_x}
    Let $C$ be the subfield of constants of $K$. We assume:
    \begin{enumerate}
      \item $[K:C]=p$.
      \item There exists $x\in K$ such that $\partial(x)=1$.
    \end{enumerate}
  \end{Hypothesis}
  \begin{Remarque}
    Since we will work on operators in $K\langle\partial\rangle$,
    $\partial$ will denote both a formal operator and a derivation on $K$.
    For the sake of simplicity we will write
    \[f':=\partial(f)\text{ and }f^{(k)}:=\partial^k(f)\]
    for any $f\in K$.
  \end{Remarque}
  We recall here a few properties of the ring of linear differential
  operators $K\langle\partial\rangle$ which are classic results on Ore
  polynomials.
  \begin{Definition-Proposition}\label{def_lclm_gcrd}
    \begin{enumerate}[label=\roman*)]
      \item Let $L_1,L_2\in K\langle\partial\rangle$. There exists a unique
        couple $Q,R\in K\langle\partial\rangle$ with
        $\mathrm{ord}(R)<\mathrm{ord}(L_2)$ such that $L_1=QL_2+R$.
      \item Every left ideal of $K\langle\partial\rangle$ is principal.
      \item Let $L_1,\dots,L_n\in K\langle\partial\rangle$. We denote by
        $\mathrm{LCLM}(L_1,\dots,L_n)$ the unique monic generator of
        $\bigcap_{i=1}^n K\langle\partial\rangle L_i$.
      \item Let $L_1,\dots,L_n\in K\langle\partial\rangle$. We denote by
        $\mathrm{GCRD}(L_1,\dots,L_n)$ the unique monic generator of the
        left ideal generated by $L_1,\dots L_n$.
      \item Let $L_1,L_2\in K\langle\partial\rangle$.
        $$\mathrm{ord}(\mathrm{LCLM}(L_1,L_2))+\mathrm{ord}(\mathrm{GCRD}(L_1,L_2))=\mathrm{ord}(L_1)+\mathrm{ord}(L_2).$$
    \end{enumerate}
  \end{Definition-Proposition}
  It should be noted that this equivalent definitions exists for right
  ideals. However we won't need to use them.\par
  The right euclidean division can be used to compute $\mathrm{GCRD}$ the
  same way Euclidean algorithm is used to compute $\mathrm{gcd}$'s in the
  commutative setting. With a little more work, this algorithm can be
  adapted to compute $\mathrm{LCLM}$'s. Faster algorithms have been
  developed for those operations, for example in~\cite{BCSZ12,Gri90}. It
  is enough for now to know that those operations are computationally
  available.
  \begin{Notation}
    Let $L\in K\langle\partial\rangle$. We denote
    \[\Dcal_L:=\nicefrac{K\langle\partial\rangle}{K\langle\partial\rangle
    L}\]
    and for any right divisor $L_*\in K\langle\partial\rangle$, 
    \[\Dcal_LL_*:=\nicefrac{(K\langle\partial\rangle
    L_*+K\langle\partial\rangle L)}{K\langle\partial\rangle
    L}\]
  \end{Notation}
  \begin{Lemme}
    Let $L,L_*\in K\langle\partial\rangle$. $\Dcal_L
    L_*=\nicefrac{K\langle\partial\rangle
    \mathrm{GCRD(L,L_*)}}{K\langle\partial\rangle L}$.
  \end{Lemme}
  \begin{proof}
    This is a direct consequence of the definition of
    $\mathrm{GCRD}(L,L_*)$.
  \end{proof}
  The quotient module $\Dcal_L$ is important because of its relation to the
  factors of $L$.
  \begin{Proposition}\label{prop_bij_submodule_divisor}
    The map
    \[L_*\mapsto \Dcal_L L_*\]
    induces a bijection between the set of monic right divisors of $L$ and the
    set of submodules of $\Dcal_L$.
  \end{Proposition}
  \begin{proof}
    Let $M$ be a submodule of $\Dcal_L$. Then the kernel of the map
    $K\langle\partial\rangle\rightarrow\nicefrac{\Dcal_L}{M}$ is a left
    ideal of $K\langle\partial\rangle$ therefore it has a unique monic
    generator $L_M$ and is equal to $K\langle\partial\rangle L_M$. Since
    $L$ is an element of this kernel it follows that $L_M$ is a monic
    right divisor of $L$.\\
    Since $K\langle\partial\rangle\rightarrow \Dcal_L$ is
    surjective and $\nicefrac{\Dcal_L}{\Dcal_L L_*}\simeq\Dcal_{L_*}$ it
    follows that the maps $M\mapsto L_M$ and
    $L_*\mapsto\Dcal_L L_*$ are inverse of one another.
  \end{proof}
  As we previously mentioned, we will restrict our study to operators $L\in
  K\langle\partial\rangle$ for which there exists an irreducible polynomial
  $N\in C[Y]$ such that $L$ is a divisor of $N(\partial^p)$. One way to
  bring ourselves back to this case is through the following lemma:

  \begin{Lemme}
    Let $\psi^L_p$ be the $K$-endomorphism of
    $\Dcal_L$ corresponding to the multiplication by $\partial^p$. Its
    characteristic polynomial $\chi(\psi^L_p)$ lives in $C[Y]$. Let $N\in C[Y]$ be an
    irreducible divisor of $\chi(\psi^L_p)$. Then
    $\mathrm{GCRD}(N(\partial^p),L)\neq 1$.
  \end{Lemme}
  \begin{proof}
    We first show that $\psi^L_p$ is a $K\langle\partial\rangle$-endomorphism of $\Dcal_L$. Indeed,
    according to Hypothesis~\ref{hypothesis_x}, any element of $K$ can be
    written as $\sum_{i=0}^{p-1}c_i x^i$ with the $c_i\in C$, from which it
    follows that $\partial^p(K)=0$. Then for any $f\in K$ we find that
    $\partial^p f=f\partial^p+f^{(p)}=f\partial^p$, therefore $\partial^p$
    commutes with all the elements of $K$, and since it commutes with
    $\partial$ it is an element of the center of $K\langle\partial\rangle$.
    The multiplication by a central element of a ring is always an
    endomorphism on the ring's modules.\\
    Let's show that $\psi^L_p$ has its minimal polynomial
    $\chi_{\mathrm{min}}(\psi^L_p)$ in $C[Y]$, which also shows that
    $\chi(\psi^L_p)\in C[Y]$. Let
    $\chi_{\min}(\psi^L_p)(Y)=\sum_{i=0}^{r}l_i Y^i$ with $l_i\in K$ and
    $l_r=1$.
    Since $\psi^L_p$ is a $K\langle\partial\rangle$-endomorphism, it
    commutes with the multiplication by $\partial$ on $\Dcal_L$ so for all
    $f\in\Dcal_L$ we have
    \begin{align*}
      \chi_{\mathrm{min}}(\psi^L_p)(\partial\cdot f)&=0\\
      &=\sum_{i=0}^r
      l_i(\psi^L_p)^i(\partial\cdot f)\\
      &=\sum_{i=0}^r l_i \partial\cdot (\psi^L_p)^i(f)\\
      &=\partial\cdot\underbrace{\big(\sum_{i=0}^rl_i(\psi^L_p)^i(f)\big)}_{=\chi_\mathrm{min}(\psi^L_p)(f)=0}-\sum_{i=0}^r
      l_i'(\psi^L_p)^i(f)\\
      &=-\sum_{i=0}^r
      l_i'(\psi^L_p)^i(f)
    \end{align*}
    It follows that $\chi_\mathrm{min}(\psi^L_p)|\sum_{i=0}^r l_i' Y^i$. But since
    $l_r=1$, this is only possible if $l_i'=0$ for all $i$.\\

    Let $Q\in C[Y]$ be such that $NQ=\chi_{\mathrm{min}}(\psi^L_p)$. If $\deg Q=0$ then up
    to a multiplicative constant, this means that $\chi_\mathrm{min}(\psi^L_p)=N$. In
    particular, $N(\psi^L_p)(1)=N(\partial^p)\mod L=0$. It follows that $L$
    is a divisor of $N(\partial^p)$ and
    $\mathrm{GCRD}(N(\partial^p),L)=L$.
    If $\deg Q\neq 0$ then $\ker Q(\psi^L_p)$ is a proper submodule of
    $\Dcal_L$ therefore there exists $L_*\neq 1$ a divisor of $L$ such that
    $\ker Q(\psi^L_p)=\Dcal_L L_*$. But then $\ker
    Q(\psi^L_p)=N(\psi^L_p)(\Dcal_L)=\nicefrac{K\langle\partial\rangle
    L+K\langle\partial\rangle
    N(\partial^p)}{K\langle\partial\rangle}=\Dcal_L\mathrm{GCRD}(L,N(\partial^p))$.

    It follows that $\mathrm{GCRD}(L,N(\partial^p))=L_*\neq 1$.
  \end{proof}

  For the rest
  of this section we suppose that $N$ is fixed.
  \begin{Notation}\label{notation_K_N}
    We denote by $C_N:=\nicefrac{C[Y]}{N(Y)}$ the splitting field of $N$ over
    $C$ and by $y_N$ the image of $Y$ in $C_N$.
    We also set $K_N=K[y_N]$.
  \end{Notation}
  \begin{Proposition}\label{properties_sec_1}
    \begin{enumerate}[label=\roman*)]
      \item For any $f\in K$, $f^{(p)}=0$.
      \item $K\langle\partial\rangle$ is a free algebra of dimension $p^2$
        over its center $C[\partial^p]$.
      \item $\Dcal_{N(\partial^p)}$ is a central simple $C_N$-algebra of
        dimension $p^2$ (by identifying $C_N:=\nicefrac{C[Y]}{N(Y)}$ with
        $\nicefrac{C[\partial^p]}{N(\partial^p)}\subset\Dcal_{N(\partial^p)}$).
      \item $\Dcal_{N(\partial^p)}$ is either a division algebra or is
        isomorphic to $M_p(C_N)$.
      \item If $N(\partial^p)$ is a division algebra then $N(\partial^p)$
        is irreducible. If $\Dcal_{N(\partial^p)}$ is isomorphic to
        $M_p(C_N)$ then all irreducible divisors of $N(\partial^p)$ are of
        order $\deg(N)$.
    \end{enumerate}
  \end{Proposition}
  \begin{proof}
      $(i)$ and $(ii)$ are done in~\cite[Lemma~1.1]{Put95}. $(iii)$
      is~\cite[Lemma~1.2]{Put95}. Wedderburn's
      Theorem~\cite[Theorem~2.1.3]{GiSz06} states that, as a central simple
      $C_N$-algebra, $\Dcal_{N(\partial^p)}$ must be of the form $M_n(D)$
      where $D$ is itself a central simple $C_N$-division algebra. In this
      case we must have
      $\dim_{C_N}\Dcal_{N(\partial^p)}=p^2=n^2\dim_{C_N}D$. Since $p$ is a
      prime number, we either have $n=p$ in which case $D$ must be equal to
      $C_N$ or $n=1$ in which case $\Dcal_{N(\partial^p)}$ is a division
      algebra, which proves $(iv)$.\\
      Suppose that $\Dcal_{N(\partial^p)}$ is a division algebra. Then in
      particular it has no nontrivial zero divisor. Thus $N(\partial^p)$
      has no nontrivial divisor.\\
      Suppose now that $\Dcal_{N(\partial^p)}$ is isomorphic to
      $M_p(C_N)$ and let $L$ be an irreducible divisor of $N(\partial^p)$.
      Then $\Dcal_L$ is a simple $K\langle\partial\rangle$-module. We can
      apply~\cite[Proposition~1.7.1]{Put95} to $\Dcal_L$. Since $\Dcal_L$
      is simple we must have $\Dcal_L\simeq I(N)$ (with the notations of~\cite[Proposition~1.7.1]{Put95}). 
      In particular, $\dim_{C_N}\Dcal_L=p$ and
      $\mathrm{ord}(L)=\dim_K\Dcal_L=\frac{[C_N:C]}{[K:C]}\dim_{C_N}\Dcal_L=\deg(N)$.
  \end{proof}
  \begin{Remarque}
    $(iv)$ was also done in~\cite[Corollary~1.3]{Put95}
  \end{Remarque}
  \begin{Lemme}
    If $N$ is not separable over $C$ then $\Dcal_{N(\partial^p)}\simeq
    M_p(C_N)$.
  \end{Lemme}
  \begin{proof}
    If $N$ is not separable then there exists $N_*\in K[Y]$ such that
    $N(Y)=N_*^p(Y)$. Therefore $N_*(\partial^p)$ is a non trivial divisor
    of $N(\partial^p)$ (because $\partial^p$ commutes with the elements of
    $K$ therefore $(N_*(\partial^p))^p=N_*^p(\partial^p)=N(\partial^p)$) so
    $\Dcal_{N(\partial^p)}$ cannot be a division algebra.
  \end{proof}
  We now show how the $p$-Riccati equation appears when $N$ is a polynomial
  of degree $1$.
  \begin{Proposition}\label{p-riccati_order_one}
    Let $a\in C$. If $\partial^p-a$ is not irreducible then its monic irreducible
    divisors are the operators of the form $\partial-f$ with $f$ verifying
    \[f^{(p-1)}+f^p=a\]
  \end{Proposition}
  \begin{Remarque}
    It should be noted that as $\partial^p-a$ is a central operator in
    $K\langle\partial\rangle$, any of its right divisor is also a left
    divisor and vice-versa. Thus, for central operators we will not need to
    specify the ``side'' of the divisors. 
  \end{Remarque}
  \begin{proof}
    Let us suppose that $\partial^p-a$ is not irreducible. Then
    $\Dcal_{\partial^p-a}$ is isomorphic to $M_p(C)$. Let $L$ be a
    monic irreducible divisor of $\partial^p-a$. From
    Proposition~\ref{properties_sec_1} (v) we know that $L$ is of order $1$
    so it is of the form $\partial-b$. We consider the $K$-linear
    endomorphism $\psi^L_p$ of $\Dcal_L$ given by $\psi^L_p:M\mapsto \partial^p\cdot M$. Since
    $\partial^p$ is central in $K\langle\partial\rangle$, this is indeed a
    $K$-linear map. Furthermore, $\Dcal_L$ is isomorphic to $K$ as a
    $K$-vector space so there exists $g\in K$
    such that $\psi^L_p$ is the multiplication by $g$. Then we have
    \begin{align*}
      \partial g &\equiv\partial\cdot \partial^p\pmod L\\
      &\equiv\partial^p\cdot \partial\pmod L\\
      &\equiv g\partial\pmod L
    \end{align*}
    Thus $\partial g-g\partial=0\pmod L$. Since $\partial
    g-g\partial=g'$, it follows that $g'=0$ and $g\in C$. Moreover, $Y-g$ is the
    characteristic polynomial of $\psi^L_p$, so $\psi^L_p-g\Id=0$. In
    particular $\partial^p-g=0\pmod L$. Thus $L$ is a common divisor of both
    $\partial^p-a$ and $\partial^p-g$. This is possible only if $g=a$.
    According to~\cite[Lemma~1.4.2]{Put95}, $g=b^{(p-1)}+b^p$.\\
    Conversely if $L$ is of the form $\partial-b$ with $b^{(p-1)}+b^p=a$ then
    from what precedes it is a divisor of
    $\partial^p-b^{(p-1)}-b^p=\partial^p-a$.
  \end{proof}
  It follows that $\partial^p-a$ is irreducible in
  $K\langle\partial\rangle$ if and only if the equation
  $f^{(p-1)}+f^p=a$ has no solution in $K$. We now extend this result to
  separable polynomials $N$ of higher degree.
  \begin{Proposition}\label{prop_split_riccati}
    We assume $N$ to be separable over $C$.
    \begin{enumerate}[label=\roman*)]
      \item $K_N$ verifies Hypothesis~\ref{hypothesis_x} and
        $[K_N:K]=\deg(N)$.
      \item The canonical morphism $\Dcal_{N(\partial^p)}\rightarrow
        \nicefrac{K_N\langle\partial\rangle}{(\partial^p-y_N)}$ is an
        isomorphism of $C_N$-algebras.
      \item $\Dcal_{N(\partial^p)}$ is isomorphic to $M_p(C_N)$ if and only
        if
        the $p$-Riccati equation relative to $N$,
        \[f^{(p-1)}+f^p=y_N,\]has a solution in $K_N$.
    \end{enumerate}
  \end{Proposition}
  \begin{proof}
    \begin{enumerate}[label=\roman*)]
      \item Let $x\in K$ be such that $x'=1$. Since $[K:C]=p$, we have
        $K=C[x]$. Furthermore, noticing that $x^p\in C$, we find that $Y^p-x^p$ is the minimal
        polynomial of $x$ over $C$. In particular, since $N$ is supposed to
        be separable, $x\notin C_N$. Thus we also have $K_N=C_N[x]$.
        Furthermore, since the minimal polynomial of $x$ over $C$ is
        inseparable, that must also be the case of its minimal polynomial
        over $C_N$. Thus we have $[K_N:C_N]=p$.
        Furthermore, we have $[K_N:C]=[K_N:C_N][C_N:C]=[K_N:K][K:C]$ so 
        $[K_N:K]=\frac{p\deg(N)}{p}=\deg(N)$.
      \item   Let
        $\varphi_N:\Dcal_{N(\partial^p)}\rightarrow\nicefrac{K_N\langle\partial\rangle}{(\partial^p-
        y_N)}$ be the canonical morphism. We first show that $\varphi_N$ is injective.\\
        Let $\overbar{L}\in \ker(\varphi_N)$ and $L\in
        K\langle\partial\rangle$ be a lift of $\overbar{L}$.
  We can write $L=\sum_{0\leq i,j\leq p-1}l_{i,j}(\partial^p)x^i\partial^j$ with the
        $l_{i,j}\in C[Y]$. Since $K_N$ verifies
        Hypothesis~\ref{hypothesis_x}, we deduce that the
  family $(x^i\partial^j)_{0\leq i,j\leq p-1}$ is a $C_N$-basis of
  $\nicefrac{K_N\langle\partial\rangle}{\partial^p-y_N}$. This means that
  for all $i,j\in\ldbrack0;p-1\rdbrack$, $Y-y_N$ divides $l_{i,j}$.\\
  Thus $y_N$ is a root of all $l_{i,j}$. But since the $l_{i,j}$ all have
  coefficients in $C$ and $N$ is the minimal polynomial of $y_N$ over $C$,
  it follows that $N$ divides all $l_{i,j}$.\\
  Thus the ideal generated by $N(\partial^p)$ is precisely the kernel of
  the considered map. It follows that $\varphi_N$ is injective.
  Observe finally that $\dim_{C_N}(\Dcal_{N(\partial^p)})=p^2$ and
  $\dim_{C_N}(\nicefrac{K_N\langle\partial\rangle}{(\partial^p-y_N)})=
  p \cdot [K_N:C_N]=p^2$.
  It follows that $\varphi_N$ is also surjective by equality of dimensions.
\item $\Dcal_{N(\partial^p)}$ is isomorphic to $M_p(C_N)$ if and only if
  $\nicefrac{K_N\langle\partial\rangle}{(\partial^p-y_N)}$ is isomorphic to
        $M_p(C_N)$ which to say that $\partial^p-y_N$ admits an irreducible
        divisor of order $1$ in $K_N\langle\partial\rangle$. From
        Proposition~\ref{p-riccati_order_one}, this is equivalent to the
        $p$-Riccati equation with respect to $N$ having a solution in
        $K_N$.
    \end{enumerate}
  \end{proof}

  \begin{Notation}
    If $N\in C[Y]$ is an irreducible separable polynomial, we denote by
    $S_N$ the set of elements $f\in K_N$ verifying
    \[f^{(p-1)}+f^p=y_N.\]
  \end{Notation}
  \begin{Remarque}
    The equation $f^{(p-1)}+f^p=y_N$ can be seen as a norm equation on
    $K_N\langle\partial\rangle$. Indeed, as
    $\nicefrac{K_N\langle\partial\rangle}{\partial^p-y_N}$ is a
    central simple $C_N$-algebra, it is equipped with a uniquely defined
    reduced norm
    $\Ncal:\nicefrac{K_N\langle\partial\rangle}{\partial^p-y_N}\rightarrow C_N$. Saying that
    $f\in S_N$ is equivalent to saying that $\Ncal(\partial-f)=0$.\\

    In fact it is better to say that $K_N\langle\partial\rangle$ is an
    Azumaya algebra over $C_N[\partial^p]$ (which is locally isomorphic to a
    central simple algebra for the Zariski topology) and thus itself
    equipped with a reduced norm $\Ncal$. The $p$-Riccati equation
    relative to $N$ is thus the norm equation 
    \[\Ncal(\partial-f)=\partial^p-y_N.\]
  \end{Remarque}
  \begin{Lemme}\label{affine_S_N}
    Let $N\in C[Y]$ be an irreducible separable polynomial and $f\in S_N$.
    Then $S_N=\{f-\frac{g'}{g}|g\in K_N\}$.
  \end{Lemme}
  \begin{proof}
    Let $h$ be an element of $K_N$. The element $h$ is in $S_N$ if and 
    only if $h-f$ verifies $(h-f)^{(p-1)}+(h-f)^p=0$, which is the same
    as requiring that
    $\partial-(h-f)$ is a divisor of $\partial^p$. Thus there exists $L_*$
    of order $p-1$ in $K_N\langle\partial\rangle$ such that
    $L_*(\partial-(h-f))=\partial^p$. Both $L_*$ and $\partial-(h-f)$ acts
    as $C_N$-linear maps on $K_N$ by evaluation. Thus
    $\dim_{C_N}\ker(\partial^p)\leq
    \dim_{C_N}\ker(L_*)+\dim_{C_N}\ker(\partial-(h-f))$. But since
    $\dim_{C_N}\ker(\partial^p)=p$ and $\dim_{C_N}\ker(L_*)\leq p-1$ and
    $\dim_{C_N}\ker(\partial-(h-f))\leq 1$ we must in fact have only
    equalities. In particular this means that $\partial-(h-f)$ has a
    solution $g\in K_N$ and is its minimal vanishing operator, that is to
    say is equal to $\partial-\frac{g'}{g}$. This is in fact an equivalence
    as every operator of the form $\partial-\frac{g'}{g}$ has a solution in
    $K_N$ (namely $g$) and is thus a divisor of $\partial^p$. Thus $h\in
    S_N$ if and only if $h-f$ is of the form $\frac{g'}{g}$.
  \end{proof}

  In Section~\ref{section_fact_eff}, we will also see how to effectively use the
  solutions of the $p$-Riccati equation to compute irreducible divisors of
  $N(\partial^p)$. Using those solutions to more generally compute
  irreducible divisors of some $L\in K\langle\partial\rangle$, where $L$ is
  a divisor of $N(\partial^p)$ is harder, in the sense that, as things
  stand, we are not able to avoid a linear dependency in $p$ in the
  ``size'' of the coefficients of the irreducible factor of $L$. The
  precise meaning of this statement will be explained in
  Section~\ref{section_fact_eff}.\par
  Still, from a theorical standpoint, solving the $p$-Riccati equation
  relative to $N$ is
  completely equivalent to being able to factor any divisor of
  $N(\partial^p)$. We show how in the rest of this section. From now on,
  $N$ is supposed separable over $C$ (unless stated otherwise).
  \begin{Proposition}\label{prop_factN(dp)}
    Let $\varphi_N:\Dcal_{N(\partial^p)}\rightarrow
    \nicefrac{K_N\langle\partial\rangle}{\partial^p-y_N}$ denote the
    canonical morphism. As per Proposition~\ref{prop_split_riccati}(ii),
    $\varphi_N$ is an isomorphism. Let $f\in S_N$. For any $f\in S_N$ we
    denote 
    \[\Lcal(f)=\mathrm{GCRD}(N(\partial^p),\varphi_N^{-1}(\partial-f)).\]
    \begin{enumerate}[label=\roman*)]
      \item 
        $f\mapsto \Lcal(f)$ is a bijection between $S_N$ and the proper
        irreducible monic
        divisors of $N(\partial^p)$.
      \item If $S_N\neq\varnothing$, for any $f\in S_N$
        \[N(\partial^p)=\mathrm{LCLM}_{i\in\Fbb_p}\Lcal\big(f+\frac{i}{x}\big).\]
    \end{enumerate}
  \end{Proposition}
  \begin{proof}
    \begin{enumerate}[label=\roman*)]
      \item Let $f\in S_N$. Let us first show that $\Lcal(f)$ is indeed a
        monic irreducible divisor of $N(\partial^p)$. Since $f\in S_N$, we
        know that $\partial-f$ is an irreducible divisor of
        $\partial^p-y_N$ in $K_N\langle\partial\rangle$. In particular this
        means that
        $\nicefrac{K_N\langle\partial\rangle}{K_N\langle\partial\rangle(\partial-f)}$
        is a simple $K_N\langle\partial\rangle$-module. But since 
        $\nicefrac{K_N\langle\partial\rangle}{K_N\langle\partial\rangle(\partial-f)}$
        is actually a
        $\nicefrac{K_N\langle\partial\rangle}{\partial^p-y_N}$-module and
        $K\langle\partial\rangle\twoheadrightarrow\nicefrac{K_N\langle\partial\rangle}{\partial^p-y_N}$
        is surjective, it is actually simple as a
        $K\langle\partial\rangle$-module.
        We claim that
        $\varphi_N^{-1}$ induces an isomorphism of
        $K\langle\partial\rangle$-modules between
        $\nicefrac{K_N\langle\partial\rangle}{K_N\langle\partial\rangle(\partial-f)}$
        and $\Dcal_{\Lcal(f)}$.
        This is not hard to see as
        \[
          \nicefrac{K_N\langle\partial\rangle}{K_N\langle\partial\rangle(\partial-f)}
          \simeq
        \left(\nicefrac{K_N\langle\partial\rangle}{\partial^p-y_N}\right)/
        \left(\nicefrac{K_N\langle\partial\rangle(\partial-f)}{\partial^p-y_N}\right)\]
      and
        \[\Dcal_{\Lcal(f)}\simeq
        \Dcal_{N(\partial^p)}/\Dcal_{N(\partial^p)}\Lcal(f)\]
        Since
        $\varphi_N^{-1}(\nicefrac{K_N\langle\partial\rangle}{\partial^p-y_N})=\Dcal_{N(\partial^p)}$
        it is enough to show that $\varphi_N^{-1}$ maps the left ideal of
        $\nicefrac{K_N\langle\partial\rangle}{\partial^p-y_N}$ generated by
        $\partial-f$ to the left ideal of $\Dcal_{N(\partial^p)}$ generated
        by $\Lcal(f)$. But since $\varphi_N^{-1}$ is a ring isomorphism,
        \begin{align*}
          \varphi_N^{-1}(\nicefrac{K_N\langle\partial\rangle(\partial-f)}{\partial^p-y_N})
          &=\Dcal_{N(\partial^p)}\varphi_N^{-1}(\partial-f)\\
          &=\Dcal_{N(\partial^p)}\mathrm{GCRD}(N(\partial^p),\varphi_N^{-1}(\partial-f))\\
          &=\Dcal_{N(\partial^p)}\Lcal(f)
        \end{align*}
        
        Let now $L$ be a monic irreducible divisor of $N(\partial^p)$. Since
        $\Dcal_L\simeq
        \nicefrac{\Dcal_{N(\partial^p)}}{\Dcal_{N(\partial^p)}L}$ is
        simple, it follows that $\Dcal_{N(\partial^p)}$ is a maximal proper
        submodule. Then
        $\varphi_N(\Dcal_{N(\partial^p)}L)$ is also a maximal proper submodule of
        $\nicefrac{K_N\langle\partial\rangle}{\partial^p-y_N}$, which means
        that there is a unique $f\in S_N$ such that
        $\varphi_N(\Dcal_{N(\partial^p)}L)=\nicefrac{K_N\langle\partial\rangle(\partial-f)}{\partial^p-y_N}$.
        Indeed, the quotient by $\varphi_N(\Dcal_{N(\partial^p)}L)$ is
        simple so it must be generated by an irreducible divisor of
        $N(\partial^p)$. Then $L=\Lcal(f)$. Indeed, as previously stated,
        $\Dcal_{N(\partial^p)}\Lcal(f)=\varphi_N^{-1}(\nicefrac{K_N\langle\partial\rangle(\partial-f)}{\partial^p-
        y_N})=\varphi_N^{-1}(\varphi_N(\Dcal_{N(\partial^p)}L)=\Dcal_{N(\partial^p)}L$.
        Since $\Lcal(f)$ and $L$ are both monic divisors of $N(\partial^p)$
        and generates the same submodule, by
        Proposition~\ref{prop_bij_submodule_divisor}, they must be equal.\\

        Finally, if $f_1,f_2\in S_N$ are such that $\Lcal(f_1)=\Lcal(f_2)$
        then
        $\varphi_N^{-1}(\nicefrac{K_N\langle\partial\rangle(\partial-f_1)}{\partial^p-y_N})=\Dcal_{N(\partial^p)}\Lcal(f_1)=\Dcal_{N(\partial^p)}\Lcal(f_2)=\varphi_N^{-1}(\nicefrac{K_N\langle\partial\rangle(\partial-f_2)}{\partial^p-y_N})$,
        thus $\partial-f_1$ and $\partial-f_2$ generates the same
        submodule. Since there are both monic divisors of $\partial^p-y_N$,
        by Proposition~\ref{prop_bij_submodule_divisor}, we have $f_1=f_2$.
      \item Let $f\in S_N$. Saying that
        $N(\partial^p)=\mathrm{LCLM}_{i\in\Fbb_p}\Lcal\big(f+\frac{i}{x}\big)$
        is by definition the same as saying that $K\langle\partial\rangle
        N(\partial^p)=\bigcap_{i\in\Fbb_p}K\langle\partial\rangle\Lcal\big(f+\frac{i}{x}\big)$.
        This is equivalent to saying that
        \begin{align*}
          &\bigcap_{i\in\Fbb_p}\Dcal_{N(\partial^p)}\Lcal\big(f+\frac{i}{x}\big)=\{0\}\\
          \Leftrightarrow&\varphi_N\bigg(\bigcap_{i\in\Fbb_p}\Dcal_{N(\partial^p)}\Lcal\big(f+\frac{i}{x}\big)\bigg)=\{0\}\\
          \Leftrightarrow&\bigcap_{i\in\Fbb_p}\varphi_N\bigg(\Dcal_{N(\partial^p)}\Lcal\big(f+\frac{i}{x}\big)\bigg)=\{0\}\\
          \Leftrightarrow&\bigcap_{i\in\Fbb_p}\nicefrac{K_N\langle\partial\rangle(\partial-f-\frac{i}{x})}{
            \partial^p-y_N}=\{0\}\\
          \Leftrightarrow&\mathrm{LCLM}_{i\in\Fbb_p}(\partial-f-\frac{i}{x})=\partial^p-y_N
        \end{align*}
        We introduce the automorphism of $K_N\langle\partial\rangle$,
        $\tau_f:\partial\mapsto \partial+f$. Since $\tau_f$ is a ring
        autmorphism, it commutes with $\mathrm{LCLM}$ operations.
        Furthermore it must map central elements to central elements and
        preserve divisibility relations. Since $\partial^p-y_N$ is the
        smallest central multiple of $\partial-f$, $\tau_f(\partial^p-y_N)$
        must be the smallest central multiple of
        $\tau_f(\partial-f)=\partial$, thus
        $\tau_f(\partial^p-y_N)=\partial^p$. Thus
        \begin{align*}
          &\mathrm{LCLM}_{i\in\Fbb_p}(\partial-f-\frac{i}{x})=\partial^p-y_N\\
          \Leftrightarrow&\mathrm{LCLM}_{i\in\Fbb_p}\tau_f(\partial-f-\frac{i}{x})=\tau_f(\partial^p-y_N)\\
          \Leftrightarrow&\mathrm{LCLM}_{i\in\Fbb_p}(\partial-\frac{i}{x})=\partial^p
        \end{align*}
        But $x^i$ is a solution of the operator $\partial-\frac{i}{x}$
        which means that $\mathrm{LCLM}_{i\in\Fbb_p}\partial-\frac{i}{x}$
        must vanish on $K_N$. Since its order is at most $p$, it must be
        equal to $\partial^p$ and we get the result.
    \end{enumerate}
  \end{proof}
  \begin{Proposition}\label{prop_factfactors}
    Let $L$ in $K\langle\partial\rangle$ be a divisor of $N(\partial^p)$
    and let $R\in K\langle\partial\rangle$be such that $LR=N(\partial^p)$.
    Let $f\in S_N$.
    There exists $n\in\ldbrack 0,p\rdbrack$ such $\mathrm{ord}(L)=n\deg(N)$
    and $J\subset\Fbb_p$ such that
    \begin{enumerate}[label=\roman*)]
      \item $\#J=n$
      \item For all $i\in J$,
        $\mathrm{LCLM}\big(R,\Lcal\big(f+\frac{i}{x}\big)\big)\cdot R^{-1}$ is an
        irreducible right divisor of $L$.
      \item $L=\mathrm{LCLM}_{i\in
        J}\bigg(\mathrm{LCLM}\big(R,\Lcal\big(f+\frac{i}{x}\big)\big)\cdot
        R^{-1}\bigg)$.
    \end{enumerate}
  \end{Proposition}
  \begin{proof}
    Let us first observe that for all $i\in\Fbb_p$,
    $\mathrm{LCLM}\big(R,\Lcal\big(f+\frac{i}{x}\big)\big)\cdot R^{-1}$ is either equal
    to $1$ if $\Lcal(f+\frac{i}{x})$ divides $R$, or is an irreducible
    divisor of $L$. The first part of the statement is obvious, so let us
    assume that $\Lcal\big(f+\frac{i}{x}\big)$ does not divide $R$.
    Replacing $f$ by $f+\frac{i}{x}$ which is another element of $S_N$, we
    can assume $i=0$. Since $\Lcal(f)$ is irreducible and does not divide
    $R$, $\mathrm{GCRD}(R,\Lcal(f))=1$. From
    Proposition~\ref{def_lclm_gcrd}($v$) it follows that
    $\mathrm{ord}(\mathrm{LCLM}(R,\Lcal(f))=\mathrm{ord}(R)+\deg(N)$. Thus
    $\mathrm{ord}(\mathrm{LCLM}(R,\Lcal(f))\cdot R^{-1})=\deg(N)$.
    Furthermore, by definition of $\mathrm{LCLM}(R,\Lcal(f))$, we have
    $\Dcal_{N(\partial^p)}\mathrm{LCLM}(R,\Lcal(f))=\Dcal_{N(\partial^p)}R\cap\Dcal_{N(\partial^p)}\Lcal(f)$
    is a proper submodule of $\Dcal_{N(\partial^p)}R$.
    The map $M\mapsto MR$ is $K\langle\partial\rangle$-module isomorphism from $\Dcal_L$ to
    $\Dcal_{N(\partial^p)}R$. Since $\mathrm{LCLM}(R,\Lcal(f))$ generates a
    proper submodule of $\Dcal_{N(\partial^p)}R$,
    $\mathrm{LCLM}(R,\Lcal(f))\cdot R^{-1}$ generates a proper submodule of
    $\Dcal_L$. It follows that its $\mathrm{GCRD}$ with $L$ is nontrivial
    and is a divisor of $N(\partial^p)$ (since $L$ itself is a divisor of
    $N(\partial^p)$) of order smaller than $\deg(N)$. This is only possible
    if $\mathrm{LCLM}(R,\Lcal(f))\cdot R^{-1}$ itself is a divisor of $L$
    since by Proposition~\ref{properties_sec_1}($iv$) all irreducible
    divisor of $N(\partial^p)$ are of order $\deg(N)$. In particular this
    means that $\mathrm{LCLM}(R,\Lcal(f))\cdot R^{-1}$ is irreducible.\\
    
    Now, since $L$ can be written as a product of irreducible divisors of
    $N(\partial^p)$, it follows that the order of $L$ is a multiple of
    $\deg(N)$. Let $n=\mathrm{ord}(L)/\deg(N)$. Let us show, by recurrence
    on $n$, that there exists a subset $J\subset\Fbb_p$ of cardinal $n$
    such that 
    \[\Dcal_{N(\partial^p)}R\cap\bigg(\bigcap_{i\in
    J}\Dcal_{N(\partial^p)}\Lcal\big(f+\frac{i}{x}\big)\bigg)=\{0\}.\]
    Since $N(\partial^p)=\mathrm{LCLM}_{i\in\Fbb_p}
    \Lcal\big(f+\frac{i}{x}\big)$ it follows that
    $\bigcap_{i\in\Fbb_p}\Dcal_{N(\partial^p)}\Lcal\big(f+\frac{i}{x}\big)=\{0\}$.
    Thus there exists $i\in\Fbb_p$ for which
    $\Dcal_{N(\partial^p)}\Lcal\big(f+\frac{i}{x}\big)$ does not contain
    $\Dcal_{N(\partial^p)}R$, which is to say that
    $\Lcal\big(f+\frac{i}{x}\big)$ does not divide $R$. Thus
    $\mathrm{ord}\big(\mathrm{LCLM}\big(R,\Lcal\big(f+\frac{i}{x}\big)\big)\big)=\mathrm{ord}(R)+\deg(N)$
    and there exists $L_*\in K\langle\partial\rangle$ of order
    $(n-1)\deg(N)$ such that
    $L_*\mathrm{LCLM}\big(R,\Lcal\big(f+\frac{i}{x}\big)\big)=N(\partial^p)$.
    By induction hypothesis, there exists $J'\subset\Fbb_p$ of cardinal
    $n-1$ such that 
    \[\Dcal_{N(\partial^p)}\mathrm{LCLM}\big(R,\Lcal\big(f+\frac{i}{x}\big)\big)\cap\bigg(\bigcap_{j\in
    J'}\Dcal_{N(\partial^p)}\Lcal\big(f+\frac{j}{x}\big)\bigg)=\{0\}.\]
    We claim that $J'$ does not contain $i$. Indeed, if such was the case
    we would have
    \begin{align*}
      &\Dcal_{N(\partial^p)}\mathrm{LCLM}\big(R,\Lcal\big(f+\frac{i}{x}\big)\big)\cap\bigg(\bigcap_{j\in
      J'}\Dcal_{N(\partial^p)}\Lcal\big(f+\frac{j}{x}\big)\bigg)\\
      &=\Dcal_{N(\partial^p)}R\cap \Dcal_{N(\partial^p)}\Lcal\big(f+\frac{i}{x}\big)\cap\bigg(\bigcap_{j\in
    J'}\Dcal_{N(\partial^p)}\Lcal\big(f+\frac{j}{x}\big)\bigg)\\
      &=\Dcal_{N(\partial^p)}R\cap\bigg(\bigcap_{j\in
    J'}\Dcal_{N(\partial^p)}\Lcal\big(f+\frac{j}{x}\big)\bigg)\\
      &=0
      \end{align*}
      and thus $N(\partial^p)$ would be the $\mathrm{LCLM}$ of $R$ and the
      $\Lcal\big(f+\frac{j}{x}\big)$ for $j\in J'$. But since $R$ is of
      order $(p-n)\deg(N)$ that would mean that $N(\partial^p)$ is of order
      less than $(p-1)\deg(N)$ which is impossible. Thus we can take
      $J=J'\cup\{i\}$.\\

      It follows that
      \[\bigcap_{i\in J}
      \Dcal_{N(\partial^p)}\mathrm{LCLM}\big(R,\Lcal\big(f+\frac{i}{x}\big)\big)=\{0\}.\]
      Since $M\mapsto MR$ is an isomorphism from $\Dcal_L$ to
      $\Dcal_{N(\partial^p)}R$, we deduce that
      \[\bigcap_{i\in J}
      \Dcal_{L}\mathrm{LCLM}\big(R,\Lcal\big(f+\frac{i}{x}\big)\big)\cdot
      R^{-1}=\{0\}.\]
      and ($iii$) follows from the fact that
      $\mathrm{LCLM}\big(R,\Lcal\big(f+\frac{i}{x}\big)\big)\cdot R^{-1}$ is a
      divisor of $L$.
  \end{proof}
  From Proposition~\ref{prop_factfactors} we see how solving the
  $p$-Riccati equation relative to $N(\partial^p)$ is enough to be able to
  factorize any divisor of $N(\partial^p)$. As a matter of fact, the
  inverse of the map
  $f\mapsto \Lcal(f)$ from Proposition~\ref{prop_factN(dp)} can be
  expressed in a similar manner as $L_*\mapsto
  \partial-\mathrm{GCRD}(\partial^p-y_N,\varphi_N(L_*))$. Thus we see that
  knowing some irreducible factor of $N(\partial^p)$ is also enough to
  solve the $p$-Riccati equation relative to $N$. From a computational
  standpoint however, the costs of the computations of both maps are not
  equivalent. For reasons we shall develop in
  section~\ref{section_fact_eff}, computing $\Lcal(f)$, given a certain
  $f$, is an operation of cost polynomial in the bidegree of $N$ and the
  ``size'' of $f$, whereas, given a certain irreducible divisor $L_*$ of
  $N(\partial^p)$, the computation of the corresponding element of $S_N$
  yields a result of ``size'' polynomial in the bidegree of $N$ and the
  degree of $L_*$, but also linear in $p$.
\stoptoc
  \subsection*{A word of the unseperable case}
  We know that when $N$ is inseparable then $N(\partial^p)$ has a divisor
  in $N_*\in K[\partial^p]$ such that $N_*^p=N(\partial^p)$ which proves
  that $\Dcal_{N(\partial^p)}\simeq M_p(C_N)$. We can use this $N_*$ to
  recover irreducible divisors of any divisor $L\in
  K\langle\partial\rangle$ of $N(\partial^p)$.
  \begin{Proposition}
    Let $N\in C[Y]$ be an irreducible unseparable polynomial over $C$. Let
    $L\in K\langle\partial\rangle\backslash K$ be a divisor of $N(\partial^p)$ and
    $N_*\in K[\partial^p]$ such that $N_*^p(\partial^p)=N(\partial^p)$. Then
    there exists a unique $i\in\ldbrack 1;p\rdbrack$ such that
    $\mathrm{GCRD}(L,N_*^i(\partial^p))$ is an irreducible divisor of $L$.
  \end{Proposition}
  \begin{proof}
    Since
    $\{0\}=\Dcal_{N(\partial^p)}N_*^p(\partial^p)\subsetneq\Dcal_{N(\partial^p)}N_*^{p-1}(\partial^p)\subsetneq\dots
    \subsetneq
    \Dcal_{N(\partial^p)}N_*^{0}(\partial^p)=\Dcal_{N(\partial^p)}$,
    there exists a smallest $i\in \ldbrack 1;p\rdbrack$ such that
    $\Dcal_{N(\partial^p)}N_*^i(\partial^p)+\Dcal_{N(\partial^p)}L\neq\Dcal_{N(\partial^p}$.
    Then
    $M=\Dcal_{\mathrm{GCRD}(L,N_*^i(\partial^p)}\simeq\nicefrac{\Dcal_{N(\partial^p)}}{\Dcal_{N(\partial^p)}N_*^i(\partial^p)+
    \Dcal_{N(\partial^p)}L}$ is a simple $K\langle\partial\rangle$-module.\\

    Indeed by hypothesis on $i$ we know that
    $\Dcal_{N(\partial^p}=\Dcal_{N(\partial^p)}L+\Dcal_{N(\partial^p)}N_*^{i-1}(\partial^p)$
    so there is a surjective homomorphism
    \[\Dcal_{N(\partial^p)}N_*^{i-1}(\partial^p)\twoheadrightarrow M\]
    which factors into
    \[\nicefrac{\Dcal_{N(\partial^p)}N_*^{i-1}(\partial^p)}{\Dcal_{N(\partial^p)}N_*^{i}(\partial^p)}\simeq
    \Dcal_{N_*(\partial^p)}\twoheadrightarrow M.\]
    But since $\Dcal_{N_*(\partial^p)}$ is a simple
    $K\langle\partial\rangle$-module (because $\mathrm{ord} N_*(\partial^p)=\deg N$
    so $N_*(\partial^p)$ is irreducible), this morphism is also injective
    which proves that $M$ is simple and
    $\mathrm{GCRD}(L,N_*^i(\partial^p))$ is an irreducible divisor of
    $L$.\\

    By hypothesis on $i$, for all $j<i$ we have
    $\mathrm{GCRD}(L,N_*^j(\partial^p))=1$ and by the same procedure we
    show that
    $\nicefrac{\Dcal_{\mathrm{GCRD}(L,N^{l+1}_*(\partial^p))}}{\Dcal_{\mathrm{GCRD}(L,N^l_*(\partial^p)}}$
    is isomorphic to $\Dcal_{N_*(\partial^p)}$ for $l\geq i$, which in
    particular proves that
    $\Dcal_{\mathrm{GCRD}(L,N^{l+1}_*(\partial^p))}$ is not
    simple and thus the unicity of $i$.
  \end{proof}  
\resumetoc
\section{Polynomial time irreducibility
test}\label{section_irreducibility_test}
We now present an irreducibility test for operators of the form
$N(\partial^p)$ with $N$ being an irreducible polynomial over $C$. 
We restrict ourselves to the case where $K$ is a finite separable
extension of $\Fbb_p(x)$.

\begin{Lemme}\label{lemma_struc_K}
  If $K$ is a finite separable field extension of $\Fbb_p(x)$ equipped with
  the derivation $\frac{\ud}{\ud x}$  then $K$
  satisfies Hypothesis~\ref{hypothesis_x}. Furthermore its constants are the
  $p$-th powers of elements of $K$.
\end{Lemme}
\begin{proof}
  Let $C$ be the field of constants of $K$.
  $\Fbb_p(x)$ does satisfy Hypothesis~\ref{hypothesis_x} and its field of
  constant is $\Fbb_p(x^p)$. Since $K$ is a finite separable extension of
  $\Fbb_p(x)$, there exists $F\in\Fbb_p[x,Y]$ irreducible and a root $r$ of $F$
  in a separable closure of $\Fbb_p(x)$ such that $K=\Fbb_p(x)[r]$. But
  then $r^p\in C$, thus $\Fbb_p(x^p)[r^p]\subset C$. Let $\Phi$ denote the
  Frobenius endomorphism on $\Fbb_p(x)$. The element $r^p$ is a root of
  $\Phi(F)$ so $[\Fbb_p(x^p)[r^p]:\Fbb_p(x^p)]=[K:\Fbb_p(x)]$. We have
  $[K:\Fbb_p(x^p)]=[K:\Fbb_p(x)][\Fbb_p(x):\Fbb_p(x^p)]=[K:\Fbb_p(x^p)[r^p]][\Fbb_p(x^p)[r^p]:
  \Fbb_p(x^p)]$ which is to say that
  $[K:\Fbb_p(x^p)[r^p]]=\frac{\deg(F)p}{\deg(F)}=p$. Since
  $\Fbb_p(x^p)[r^p]\subset C\subset K$, and $K\neq C$ ($x\in K$) we have
  $C=\Fbb_p(x^p)[r^p]$ and $[K:C]=p$. Furthermore the elements of
  $C=\Fbb_p(x^p)[r^p]$ are exactly the $p$-th powers of elements of
  $K=\Fbb_p(x)[r]$.
\end{proof}

      \begin{Notation}
        Let $N$ be an irreducible polynomial over $C$. For any place
        $\Pfrak$ of $K_N$ we denote by $K_{N,\Pfrak}$ the completion of
        $K_N$ with regard to the associated valuation $\nu_{\Pfrak}$. We
        also denote $\Gcal_\Pfrak$ the residue class field of $K_N$.
        Finally we will usually use the notation $t_\Pfrak$ to refer to a
        prime element of $\Pfrak$ in $K_N$.\\
        For any place of $\Pfrak'$ of $C_N$, we denote by $C_{N,\Pfrak'}$ the
        completion of $(C_N,\nu_{\Pfrak'})$.\\
        Note that $f\mapsto f^p$ is an isomorphism between $K_N$ and $C_N$
        and induces isomorphisms between their completion fields too. Thus
        if $\Pfrak$ is a place in $K_N$ we will allow ourselves to write
        $C_{N,\Pfrak}$ for the completion of $C_N$ in the place which makes
        it isomorphic to $K_{N,\Pfrak}$ through this isomorphism.\\
        Conversely, if $\Pfrak$ is a place in $C_N$ we will allow ourselves
        to write $K_{N,\Pfrak}$ for the completion of $K_N$ in the place
        which makes it isomorphic to $C_{N,\Pfrak}$ through this
        isomorphism.\\
        For any algebraic function field $F$ we denote by $\Pbb_F$ the set
        of places of $F$ and by $\mathrm{Div}(F)$ the group of divisors of 
        $F$; we recall that it is 
        the free $\Zbb$-module generated by the elements of $\Pbb_F$. If $f$ is a
        nonzero element of $F$, we denote by $(f)$ the principal divisor of
        $f$, by $(f)_+$ its divisor of zeros and by $(f)_-$ its divisor of
        poles. If $D$ is a divisor over $F$, we write $\Lcal(D)=\{f\in
        F|(f)\geq -D\}$ for the Riemann-Roch space associated to $D$.\\
        We denote by $\Diff(K_N/K)$ (or just $\Diff(K_N)$) the different
        divisor of $K_N$ over $K$.\\
        Finally if $k$ is a field, we denote by $\mathrm{Br}(k)$ the Brauer
        group of $k$.
      \end{Notation}
      The basis for our irreducibility test is the following proposition
      \begin{Proposition}
        Let $N$ be an irreducible polynomial over $C$. Then, $N(\partial^p)$ is
        reducible in $K\langle\partial\rangle$ if and only if the
        $p$-Riccati equation
        \[f^{(p-1)}+f^p=y_N\]
        has a solution in $K_{N,\Pfrak}$ for all $\Pfrak\in\Pbb_{K_N}$.
      \end{Proposition}
      \begin{proof}
        We know from Proposition~\ref{prop_split_riccati} that $N(\partial^p)$ is reducible in
        $K\langle\partial\rangle$ if and only if
        $\Dcal_{N(\partial^p)}\simeq\DyN{K_N}{}$ is
        isomorphic to $M_p(C_N)$, which amounts to saying that
        $\DyN{K_N}{}$ vanishes in $\mathrm{Br}(C_N)$.
        We know that $D\mapsto \bigoplus_{\Pfrak\in\Pbb_{C_N}}
        D\otimes_{C_N}C_{N,\Pfrak}$ induces an injective group
        morphism~\cite[Corollary~6.5.4]{GiSz06}
        \[\mathrm{Br}(C_N)\hookrightarrow
        \bigoplus_{\Pfrak\in\Pbb_{C_N}}\mathrm{Br}(C_{N,\Pfrak}).\]
        In particular this means that $\DyN{K_N}{}$ is isomorphic to $M_p(C_N)$
        if and only if $$\DyN{K_N}{}\otimes_{C_N}C_{N,\Pfrak}$$ is isomorphic to
        $M_p(C_{N,\Pfrak})$ for all $\Pfrak\in\Pbb_{C_N}$.

        Besides we know that $\DyN{K_N}{}\otimes_{C_N}C_{N,\Pfrak}$ is
        isomorphic to $\DyN{K_{N,\Pfrak}}{}$. Thus
        $\DyN{K_N}{}$ is isomorphic to $M_p(C_N)$ if and only if
        $\DyN{K_{N,\Pfrak}}{}$ is isomorphic to $M_p(C_{N,\Pfrak})$
        for all $\Pfrak\in\Pbb_{K_N}$.

        Lastly $K_{N,\Pfrak}$ is of the form
        $\Fbb_q((t_\Pfrak))$ for $q$ some power of $p$. In particular it is
        a field verifying Hypothesis~\ref{hypothesis_x}. Thus it is
        isomorphic to $M_p(C_{N,\Pfrak})$ if and only if the equation
        \[f^{(p-1)}+f^p=y_N\] has a solution in $K_{N,\Pfrak}$.
      \end{proof}
      We now want to find a criteria for the $p$-Riccati equation relative
      to $N$ to have a solution in $K_{N,\Pfrak}$. In general, we know that
      if $t_\Pfrak$ is a prime element of $\Pfrak$ in $K_N$, then
      $K_{N,\Pfrak}\simeq \Gcal_\Pfrak((t_\Pfrak))$. This tells us that
      $K_{N,\Pfrak}$ is isomorphic to a Laurent series field. However, we
      don't always use this representation. Instead we assume that for each
      place $\Pfrak\in\Pbb_{K_N}$, we are
      given a morphism $\iota_\Pfrak:K_N\rightarrow \Gcal_\Pfrak((T))$ which is
      continuous for the topology associated to the place $\Pfrak$ and
      whose image is dense in $\Gcal_\Pfrak((T))$. We say that
      $\iota_\Pfrak$ is a parametrization of $K_{N,\Pfrak}$. Note that in
      $\Gcal_\Pfrak((T))$ we now have $\frac{\ud}{\ud x}=T'\frac{\ud}{\ud T}$
      where $T'=\iota_\Pfrak(x)^{-1}$.\\

      Over fields of Laurent series we can apply a Newton iteration to find solutions
      to a higher precision from a given seed as illustrated by the
      following proposition.
      \begin{Proposition}\label{prop_newton_iter}
        Let $f_0\in\Gcal_{\Pfrak}((T))$ and $n\in\Zbb$ be such that
        \[\frac{\ud^{p-1}}{\ud x^{p-1}}f_0+f_0^p=y_N+O(T^{pn}).\]
        We set $e_\Pfrak:=1-\nu_\Pfrak(T')$.
        There exists $f_1\in\Gcal_{\Pfrak}((T))$ such that
        $f_1=f_0+O(T^{pn+(p-1)e_\Pfrak})$ and
        \[\frac{\ud^{p-1}}{\ud
        x^{p-1}}f_1+f_1^p=y_N+O(T^{p(pn+(p-1)e_\Pfrak)}).\]
      \end{Proposition}
      \begin{proof}
        Let $g:=\frac{\ud^{p-1}}{\ud x^{p-1}}f_0+f_0^p-y_N$. For any
        $f\in\Gcal_\Pfrak((T))$, $\frac{\ud^{p-1}}{\ud x^{p-1}}f$ is
        a constant since $\frac{\ud^p}{\ud x^p}=0$. Since $y_N\in C_N$ is
        also a constant, it follows that $\frac{\ud g}{\ud x}=0$. Thus
        there exists $\Ical(f_0)\in\Gcal_\Pfrak((T))$ such that
        $\frac{\ud^{p-1}}{\ud x^{p-1}}\Ical(f_0)=g$. Furthermore, we claim
        that we can take $\Ical(f_0)$ such that $\nu_\Pfrak(\Ical(f_0))=
        pn+(p-1)e_\Pfrak$.\\
        Indeed, let $h\in\Im\left(\frac{\ud}{\ud x}\right)$ and
        $H=\sum_{k=\nu_\Pfrak(H)}^\infty h_kT^k\in\Gcal_\Pfrak((T))$ 
        such that $\frac{\ud}{\ud x}H=h$. Then we set
        $H_1:=H-\sum_{k\in\Zbb}h_{pk}T^{pk}$. We have $\frac{\ud}{\ud
        x}H_1=\frac{\ud}{\ud x} H=h$. Furthermore $p$ does not divide
        $\nu_\Pfrak(H_1)$. But we also have
        \[\frac{\ud}{\ud x}H_1=T'\frac{\ud}{\ud
        T}\left(\sum_{k=\nu_\Pfrak(H_1)}^\infty
        h_kT^k\right)=T'\sum_{k=\nu_\Pfrak(H_1)-1}(k+1)h_{k+1}T^k.\]
        It follows that
        $\nu_\Pfrak(h)=\nu_\Pfrak(H_1)-1+\nu_\Pfrak(T')$ which is to
        say that $h$ admits a primitive $H_1$ verifying
        $\nu_\Pfrak(H_1)=\nu_\Pfrak(h)+e_\Pfrak$. Applying this result
        $p{-}1$ times, we conclude that we can take $\Ical(f_0)$
        such that $\nu_\Pfrak(I(f_0))= pn+(p-1)e_\Pfrak$.

        Next, we consider $f_1:=f_0-\Ical(f_0)$. By definition
        $f_1=f_0+O(T^{pn+(p-1)e_\Pfrak})$ and
        \begin{align*}
          \frac{\ud^{p-1}}{\ud x^{p-1}}f_1 +f_1^p&=\frac{\ud^{p-1}}{\ud
          x^{p-1}}f_0+f_0^p-\frac{\ud^{p-1}}{\ud
          x^{p-1}}\Ical(f_0)-\Ical(f_0)^p\\
          &=g+y_N-g-\Ical(f_0)^p\\
          &=y_N+O(T^{p(pn+(p-1)e_\Pfrak)})
        \end{align*}
      \end{proof}
      \begin{Corollaire}\label{corollary_criteria}
        The $p$-Riccati equation relative to $N$ admits a solution in
        $K_{N,\Pfrak}=\Gcal_\Pfrak((T))$ if and only if there exists
        $f\in\Gcal_\Pfrak((T))$ such that
        \[\frac{\ud^{p-1}}{\ud
        x^{p-1}}f+f^p=y_N+O(T^{p(1-e_\Pfrak)}).\]
        In particular if $\nu_\Pfrak(y_N)\geq p{\cdot}\nu_\Pfrak(T')$ then
        the $p$-Riccati equation relative to $N$ always has a solution in
        $K_{N,\Pfrak}$.
      \end{Corollaire}
      \begin{proof}
        Let $f$ be such that $\frac{\ud^{p-1}}{\ud
        x^{p-1}}f+f^p=y_N+O(T^{p(1-e_\Pfrak)})$ and set $f_0:=1$.
        We construct a recursive sequence $(f_k)_{k\in\Nbb}\in
        K_{N,\Pfrak}^{\Nbb}$ such that the
        term $f_{k+1}$ is constructed from $f_k$ using
        Proposition~\ref{prop_newton_iter}. We set
        $n_k:=\max\{n\in\Nbb|\frac{\ud^{p-1}}{\ud
        x^{p-1}}f_k+f_k^p=y_N+O(T^{pn})\}$ and show that the
        sequence $n_k$ is strictly increasing. Indeed we have
        $n_0>-e_\Pfrak$. Thus $n_1\geq
        pn_0+(p-1)e_\Pfrak>n_0-(p-1)e_\Pfrak+(p-1)e_\Pfrak=n_0$. It follows
        that $n_1>n_0$ and $n_1>-e_\Pfrak$. By induction we show that
        $n_k>-e_\Pfrak$ for all $k$ and conclude that $n_{k+1}>n_k$ the
        same way.\\
        From Proposition~\ref{prop_newton_iter} it also follows that
        $f_k=f_l+O(T^{n_l})$ for all $k\geq l$. Thus a solution to
        the $p$-Riccati equation in $K_{N,\Pfrak}$ is given by
        $\lim_{k\to\infty}f_k$.

        Let us now suppose that $\nu_\Pfrak(y_N)\geq
        p{\cdot}\nu_\Pfrak(T')$. Then by definition of $e_\Pfrak$, the
        function $f=0$ verifies
        \[\frac{\ud^{p-1}}{\ud
        x^{p-1}}f+f^p=y_N+O(T^{p(1-e_\Pfrak)})\] so the $p$-Riccati
        equation relative to $N$ must have a solution in $K_{N,\Pfrak}$ by
        what precedes.
      \end{proof}
      Corollary~\ref{corollary_criteria} is very important because it states that for almost all (all
      except a finite number) place $\Pfrak\in\Pbb_{K_N}$, the $p$-Riccati
      equation has a solution in $K_{N,\Pfrak}$. Indeed,
      $\nu_\Pfrak(T')$ being the valuation of the divisor
      $2(x)_--\Diff(K_N)$, the only places where the existence of a
      solution is not obvious are the places where the valuation of the
      divisor $p^{-1}\cdot(y_N)+\Diff(K_N)-2(x)_-$ is negative. 
      Since the divisor $\Diff(K_N)$ is effective, those places are either poles of 
      $y_N$ or poles of $x$.\\
      We now see how to check if the $p$-Riccati equation has a solution in
      those places.

      \begin{Lemme}\label{p-1_power_derivation}
        For any $f\in K_{N,\Pfrak}$, $\frac{\ud^{p-1}}{\ud
        x^{p-1}}f=\frac{\ud^{p-1}}{\ud T^{p-1}}(T'^{p-1}f)$.
      \end{Lemme}
      \begin{proof}
        We consider the ring of differential operators
        $K_{N,\Pfrak}\langle\partial_*\rangle$ where
        $\partial_*f=f\partial_*+\frac{\ud}{\ud T}(f)$ for all $f\in
        K_{N,\Pfrak}$. We know that
        $\frac{\ud}{\ud x}=T'\frac{\ud}{\ud T}$ so we want to
        show that $(T'
        \partial_*)^{p-1}=\partial_*^{p-1}T'^{p-1}$. We know that for
        all $f\in K_{N,\Pfrak}$, $\frac{\ud^p}{\ud x^p}f=0$. It follows
        that $(T' \partial_*)^p f=\sum_{i=0}^{p}\begin{pmatrix}
        p\\i\end{pmatrix} \frac{\ud^i}{\ud
        x^i}f(T'\partial)^{p-i}=f(T'\partial_*)^p$. Thus
        $(T'\partial_*)^p$ commutes with all the elements of $K_{N,\Pfrak}$
        so it is an element of $K_{N,\Pfrak}[\partial_*^p]$ and is of the
        form $a_1\partial_*^p+a_0$. But the leading coefficients $a_1$ is
        necessarily $T'^{p-1}$ and $a_0=\frac{\ud^p 1}{\ud x^p}=0$. Thus
        we have
        $(T'\partial_*)^{p}=T'^p\partial_*^p=T'\partial_*^pT'^{p-1}$.
        We can simplify by $T'\partial_*$ on the left and get the desired
        equality.
      \end{proof}
      \begin{Theoreme}\label{theorem_system_local}
        Let $\Pfrak\in\Pbb(K_{N})$,
        $i_\Pfrak:K_N\rightarrow\Gcal_\Pfrak((T))$ be a parametrization of
        $K_{N,\Pfrak}$ and let $T'=\iota_\Pfrak(x)^{-1}$ We suppose that
        $\eta:=\nu_\Pfrak(T')-p^{-1}\cdot
        \nu_\Pfrak(y_N)>0$. Let
        $(g_0,g_1,\ldots,g_{\eta-1})\in\Gcal^{\eta}$ be the first $\eta$
        coefficients of $(T')^{p-1}$ and
        $(a_0,\ldots,a_{\eta-1})\in\Gcal^\eta$ be the first $\eta$
        coefficients of $\iota_\Pfrak(a)$ with $a$ being a $p$-th root of $y_N$. Let $\Gcal_\Pfrak=\Fbb_{p^b}$ with
        $b>0$. We identify $\Gcal_\Pfrak^\eta$ with $\Fbb_p^{b\eta}$. We set
        \[D_{p-1}(\Pfrak):=
        \begin{pmatrix}
          \raisebox{0ex}[4em][2em]{\normalfont\Large\bfseries 0}\\
          \begin{matrix}
            g_r&\cdots& g_0&&&&&\\
            \vdots&&&\ddots&&&&\\
            g_{\eta-1-p}&&\cdots&&g_0&0&\cdots&0\\
            g_{\eta-1}&&\cdots&&g_p&g_{p-1}&\cdots&g_0
          \end{matrix}
        \end{pmatrix}\]
        where $D_{p-1}(\Pfrak)\in M_{b\eta}(\Fbb_p)$ is a block matrix and the
        coefficient $g_i$ is the matrix of the multiplication by $g_i$ in
        $\Gcal_\Pfrak$. Let $\Phi$ be the diagonal block matrix in
        $M_{b\eta}(\Fbb_p)$ whose diagonal blocks are all the matrix of the
        Frobenius endomorphism over $\Gcal_\Pfrak$. Then the $p$-Riccati equation
        relative to $N$ has a solution in $K_{N,\Pfrak}$ if and only if the
        system
        \[(\Phi-D_{p-1}(\Pfrak))X=\Phi ^t(a_0,\ldots,a_{\eta-1})\]
        has a solution in $\Gcal_\Pfrak^{\eta}$.
      \end{Theoreme}
      \begin{proof}
        Let $f\in K_{N,\Pfrak}$ verifying $\frac{\ud^{p-1}}{\ud
        x^{p-1}}f+f^p=y_N$. Then $\nu_\Pfrak(f)\geq p^{-1}\nu_\Pfrak(y_N)$.
        Indeed we know that
        $\nu_\Pfrak(y_N)\geq\min(p\nu_\Pfrak(f),\nu_\Pfrak(\frac{\ud^{p-1}}{\ud
        x^{p-1}}f))$. If we had $\nu_\Pfrak(f)<p^{-1}\nu_\Pfrak(y_N)$ then
        in particular $\nu_\Pfrak(f)<\nu_\Pfrak(T')-1$. Since
        $\nu_\Pfrak\left(\frac{\ud^{p-1}f}{\ud x^{p-1}}\right)\geq
        \nu_\Pfrak(f)+(p-1)(\nu_\Pfrak(T')-1)$, we find that 
        $\nu_\Pfrak\left(\frac{\ud^{p-1}f}{\ud
        x^{p-1}}\right)>p\nu_\Pfrak(f)$. Thus we would have 
        $\nu_\Pfrak(y_N)=\min(p\nu_\Pfrak(f),\nu_\Pfrak(\frac{\ud^{p-1}}{\ud
        x^{p-1}}f))=p\nu_\Pfrak(f)$, which is a contradiction. Thus
        $\nu_\Pfrak(f)\geq p^{-1}\cdot\nu_\Pfrak(y_N)$.\\
        We set $f:=\sum_{k=0}^\infty
        f_k T^{k+p^{-1}\cdot\nu_\Pfrak(y_N)}$. We claim that 
        $X={}^t(f_0,\ldots,f_{\eta-1})$ is a solution of the system 
        $(\Phi-D_{p-1}(\Pfrak))X=\Phi \cdot {}^t(a_0,\ldots,a_{\eta-1})$. It is in fact
        enough to check that the vector $-D_{p-1}(\Pfrak)
        ^t(f_0,\ldots,f_{\eta-1})$ is the vector of the coefficients of
        $T^{pk}$ in $\frac{\ud^{p-1}f}{\ud x^{p-1}}$ for
        $k\in\ldbrack p^{-1}\cdot\nu_\Pfrak(y_N);
        \nu_\Pfrak(T')-1\rdbrack$.  We know
        that $\frac{\ud^{p-1}}{\ud x^{p-1}}(f)=\frac{\ud^{p-1}}{\ud
        T^{p-1}}(T'^{p-1}f)$ and the result follows from a
        straightforward computation.\\
        Conversely, if 
        $(\Phi-D_{p-1}(\Pfrak))X=\Phi \cdot {}^t(a_0,\ldots,a_{\eta-1})$ has a solution
        $(f_0,\ldots,f_{\eta-1})\in\Gcal_\Pfrak^\eta$ then we set
        $f=\left(\sum_{k=0}^{\eta-1}f_kT^{k}\right)T^{p^{-1}\cdot\nu_\Pfrak(y_N)}$
        and claim that $\frac{\ud^{p-1}}{\ud
        x^{p-1}}f+f^p=y_N+O(T^{p(1-e_\Pfrak)})$ which proves the
        existence of a solution according to
        Corollary~\ref{corollary_criteria}. From what precedes, we know
        that the equality is true for the coefficients of $t^{pk}$ for
        $k\in\ldbrack p^{-1}\cdot \nu_\Pfrak(y_N),-e_\Pfrak\rdbrack$.
        Furthermore, since $\frac{\ud^{p-1}}{\ud x^{p-1}}f+f^p$ and $y_N$
        are constant, all the other coefficients before
        $T^{p(1-e_\Pfrak)}$ are equal to zero which gives the
        desired result.
      \end{proof}
    
    We can now write an algorithm for testing the irreducibility of an
    operator $N(\partial^p)$ where $N$ is an irreducible polynomial over
    $C$. From Lemma~\ref{lemma_struc_K}, we know that we can take $N_*\in K[Y]$
    such that $N_*^p(Y)=N(Y^p)$ and $K_N\simeq\nicefrac{K[Y]}{N_*}$. If we
    denote by $a$ the image of $Y$ in $K_N$ then $y_N=a^p$. This
    representation is easier to manipulate (because smaller by a factor $p$
    in all generality) so we consider that the entry of our algorithm is
    the polynomial $N_*$.   
    \begin{algorithm}
      \begin{minipage}{.95\textwidth}
  \begin{flushleft}
    \emph{Input:} $N_*\in K[Y]$ a separable irreducible polynomial.\\
    \emph{Output:} Whether or not $N_*^p(\partial)$ is irreducible in
    $K\langle\partial\rangle$
  \end{flushleft}
  \BlankLine
    \begin{enumerate}
      \item Set $K_N:=K[a]=\nicefrac{K[Y]}{N}$ where $a$ is a root of $N$.
      \item Compute $\Sbb:=\Supp(a)_- \cup\Supp(x)_-$.
      \item For $\Pfrak$ in $\Sbb$ do:
        \begin{enumerate}
          \item Compute $\iota_\Pfrak:K_{N,\Pfrak}\rightarrow
            \Gcal_\Pfrak((T))$ a parametrization of $K_{N,\Pfrak}$.
          \item Compute $T'=\iota_\Pfrak(x)^{-1}$ and set
            $\eta:=\nu_\Pfrak(T')-\nu_\Pfrak(a)$.
          \item If $\eta>0$ do:
            \begin{enumerate}
              \item Compute $T'$ at relative precision $\eta$.
              \item With fast exponentiation on $T'$, compute the first 
                $\eta$ coefficients\\ $(g_0,\ldots,g_{\eta-1})$ of
                $T'^{p-1}$.
              \item Compute $(a_0,\ldots,a_{\eta-1})$ the first $\eta$
                coefficients of $\iota_\Pfrak(a)$.
              \item Check if the system
                $(\Phi-D_{p-1}(\Pfrak))X=\Phi
                ^t(a_0,\ldots,a_{\eta-1})$ defined in
                Theorem~\ref{theorem_system_local} has a solution in
                $\Gcal_\Pfrak^\eta$.
          \item If it doesn't, return \textbf{False} and stop the
            algorithm.
            \end{enumerate}
        \end{enumerate}
      \item return \textbf{True}
    \end{enumerate}
    \caption{irreducibility\_test\_outline}
    \label{algo_irreducibility}
      \end{minipage}
  \end{algorithm}
  The correctness of Algorithm~\ref{algo_irreducibility} is easily deduced
  from the discussion that precedes. However
  Algorithm~\ref{algo_irreducibility} says very little on the way
  the parametrization $\iota_\Pfrak$ in step~(3a) should be computed should be
  computed. For that matter it also doesn't explain how to compute $(a)_-$
  or $(x)_-$. One way to proceed would be to use
  \emph{OM}-factorization algorithms~\cite{GuMoNa11,BaNaSt12,PoWe22} which,
  although it was primarily conceived for prime ideal factorisation in
  number fields, can be quite easily adapted to the context of algebraic
  function fields and can be used to recover prime elements $t_\Pfrak$ of small
  ``size" in $K_N$. As \emph{OM}-factorisation algorithms can also be
  used to compute evaluations in residue fields, it enables us to compute
  Taylor expansions of $t_\Pfrak'$ and $a$ in a naive way.\\

  It turns out that proceeding this way is probably a waste of the
  informations provided by these algorithms and not very efficient.
  Instead we now show how to use the rational Puiseux
  expansions~\cite{Duval89} of $N_*$
  to solve those problems. Rational Puiseux expansions of $N_*$ are not
  always available if $\deg_Y(N_*)=\deg_Y(N)\geq p$ as some places may be
  wildly ramified. Should this problem arise, we could instead use
  Hamburger-Noether expansions of $N_*$~\cite{CaFa99}. While those have the benefit of
  always existing, there exists to our knowledge no complexity bounds for
  them. From now on we proceed under the assumption that rational Puiseux
  Expansions are available to us. It turns out that checking the existence
  of a local solution in $\Pfrak$ to the $p$-Riccati
  equation is even simpler when the place $\Pfrak$ has a rational Puiseux
  expansion.

  \begin{Lemme}\label{lemma_simplesyslocal}
    Let $\Pfrak$ be either a pole of $a$ or a places at infinity. If there
    exists $(\xi,\gamma,\tilde{y})\in\Gcal_\Pfrak^2\times\Gcal_\Pfrak((T))$
    such that
    \[\begin{array}{crcl}
      i_\Pfrak:&K_N&\rightarrow&\Gcal_\Pfrak((T))\\
      &x&\mapsto&\xi+\gamma T^{e_\Pfrak}\\
      &a&\mapsto&\tilde{y}
    \end{array}\]
    is well defined, continuous for the topology associated to $\Pfrak$ and
    dense, and $a_{\eta-1}$ is the coefficient of $T^{\nu(T')-1}$ in $\tilde{y}$ then
    the $p$-Riccati equation relative to $N$ has a local solution in
    $K_{N,\Pfrak}$ if and only if the equation
    \[x^p-\gamma^{1-p}x=a_{\eta-1}\] has a solution in $\Gcal_\Pfrak$.
  \end{Lemme}
  \begin{Remarque}
    The existence of such a parametrisation of $K_{N,\Pfrak}$ in particular
    proves that $p\nmid e_\Pfrak$. Indeed, any parametrisation of
    $K_{N,\Pfrak}$ extends the derivation of $K_N$ to $\Gcal_\Pfrak((T))$
    uniquely by continuity. However if $p|e_\Pfrak$ then we would have
    $\frac{\ud}{\ud x}x=0$ which is absurd.
  \end{Remarque}
  \begin{proof}
    If such $\xi,\gamma,\tilde{y}$ exists then the morphism of the Lemma
    gives a parametrization of $K_{N,\Pfrak}$. Let
    $\eta=\nu(T')-\nu_\Pfrak(a)=1-e_i-\nu_\Pfrak(a)$ and let
    $a_0,\dots,a_{\eta-1}$ be the first coefficients of $\tilde{y}$. Then
    according to Theorem~\ref{theorem_system_local}, the $p$-Riccati
    equation relative to $N$ has a local solution in $K_{N,\Pfrak}$ if and
    only if the system
    $(\Phi-D_{p-1}(\Pfrak))X=\Phi^t(a_0,\dots,a_{\eta-1})$ has a solution.
    But this system is block lower triangular and all the diagonal blocks
    except the last one are invertible as the matrices of the Frobenius
    endomorphism. Thus the obstruction can only occur in the last row.
    Since $T'^{p-1}=\gamma^{1-p} T^{(p-1)(1-e_\Pfrak)}$ the equation given
    by the last row is precisely the one written in the lemma.
  \end{proof}

  For the rest of this section we proceed under the assumption that
  $\deg_Y(N)=\deg_Y(N_*)<p$ and restrain ourselves to the case
  $K=\Fbb_q(x)$ with $q=p^b$. Let $N_*\in\Fbb_q[x,y]$ with $d_x=\deg_x
  N_*$ and $d_y=\deg_y N_*$. It follows that $K_N$ is a field extension of
  $\Fbb_q(x)$ of degree $d_y$. As such, any element $f\in K_N$ can be
  represented by $d$ rational functions in $\Fbb_q(x)$. 

  \begin{Notation}For any element
  $f=\frac{1}{D_f}\sum_{i=0}^{d_y-1} f_i a^i\in K_N$ such that
  $D_f,f_0,\ldots,f_{d_y-1}\in\Fbb_q[x]$ with
  $\gcd(D_f,f_0,\ldots,f_{d_y-1})=1$ we write
  \[\deg f:=\max(\deg D_f,\deg f_0,\ldots,\deg f_{d_y-1}).\]
  \end{Notation}

  \begin{algorithm}
      \begin{minipage}{.95\textwidth}
  \begin{flushleft}
    \emph{Input:} $N_*\in \Fbb_q[X,Y]$ a separable irreducible polynomial.\\
    \emph{Output:} Whether or not $N^p(\partial)$ is irreducible in
    $\Fbb_q(x)\langle\partial\rangle$
  \end{flushleft}
  \BlankLine
    \begin{enumerate}
      \item $l_c\leftarrow$ the leading coefficient of $N_*$.
      \item Factorize $l_c$ in $\Fbb_q[X]$.
      \item \textbf{For} $l$ ranging across all irreducible factors of
        $l_c$ \textbf{do}:
        \begin{enumerate}
          \item $\Kbb_l\leftarrow\nicefrac{\Fbb_q[X]}{l}$, $\xi_l\leftarrow$ the image
            of $X$ in $\Kbb_l$.
          \item $P_l\leftarrow N_*(X+\xi_l,Y)$
          \item Use~\cite{PoWe21} to compute the set
            $\{(\gamma_1 T^{e_1},\lceil
            y_1\rceil,\Kbb_1),\dots,(\gamma_{g_l}T^{e_{g_l}},\lceil
            y_{g_l}\rceil,\Kbb_{g_l})\}$ associated to $P_l$ as in Theorem~\ref{theorem_compratPuiseux}.
          \item \textbf{For} $1\leq i\leq g_l$ \textbf{do}:
            \begin{enumerate}
              \item $\eta\leftarrow 1-e_i-\nu_T(\lceil y_i\rceil)$
              \item \textbf{If} $\nu\geq 0$ \textbf{do}:
                \begin{enumerate}
                  \item Compute $a_0,\dots,a_{\eta-1}$ the first $\eta$
                    coefficients of $\tilde{y_i}\in\Kbb_i((T))$ as in
                    Theorem~\ref{theorem_compratPuiseux}(ii)
                  \item Check if the equation
                    $x^p-\gamma_i^{1-p}x=a_{\eta-1}$ has a solution in
                    $\Kbb_i$.
                  \item If it doesn't, return \textbf{False} and stop the
                    algorithm.
                \end{enumerate}
            \end{enumerate}
        \end{enumerate}
      \item $P_\infty\leftarrow X^{\deg_X(N_*)}N_*(1/X,Y)$
      \item Repeat steps (3c) and (3d) for $P_\infty$.
          \item return \textbf{True}.
    \end{enumerate}
        \caption{irreducibility\_test\_on\_$\Fbb_q(x)$}
    \label{algo_irreducibility_rational}
      \end{minipage}
  \end{algorithm}
\begin{Remarque}
  Note that we may also remove step (6) in
  Algorithm~\ref{algo_irreducibility_rational} if we instead had
  Theorem~\ref{theorem_compratPuiseux} yield all the rational Puiseux
  expansions of $P_\infty$ instead of just the non integral ones.
\end{Remarque}

  \begin{Theoreme}
    Let $N_*\in\Fbb_{p^b}[x,y]$ be an irreducible polynomial and let
    $d_x:=\deg_x(N_*)$ and $d_y:=\deg_y(N_*)$. Not counting the cost of
    univariate factorisations, if $d_y<p$,  
    Algorithm~\ref{algo_irreducibility_rational} determines whether
    $N_*^p(\partial)$ is
    irreducible or not in $\tilde{O}((d_yd_x^2+d_y^2dx+b^{\omega-1}(d_x+d_y)^\omega)b\log(p)+b^2(d_x+d_y)^2\log(p)^2)$ 
    bit operations.
  \end{Theoreme}
  \begin{proof}
    Let us first prove that the algorithm is correct. We begin by proving
    the following lemma:
    \begin{Lemme}\label{lemme_proofirredtesthm}
      Let $l$ be an irreducible factor of $l_c$. With the notations of
      Algorithm~\ref{algo_irreducibility_rational}, for all
      $i\in\ldbrack1,g_l\rdbrack$
      \[\begin{array}{rcl}
        K_{N_*}&\rightarrow& \Kbb_i((T))\\
        X&\mapsto&\gamma_i T^{e_i}+\xi_l\\
        a&\mapsto& \tilde{y_i}
      \end{array}\] is well defined, dense and induces a valuation on $K_{N_*}$
      associated to a place above $l$. The map that to $i$ associate
      this pole is surjective.
    \end{Lemme}
    \begin{proof}
      The fact that the morphisms are well-defined and induces a valuation
      on $K_{N_*}$ is obvious. To show that they are dense, let
      $i\in\ldbrack1;g_l\rdbrack$.
      There is an irreducible factor $N_i$ of $N_*(X+\xi_l,Y)$ in $\Kbb_l[X,Y]$ such
      that $N_i(\tilde{x_i},\tilde{y_i})=0$. Let $F=\Kbb_l\cdot K_{N_*}$.
      Then $F\simeq \nicefrac{\Kbb_l(X)[Y]}{N_i}$ through the morphism
      $(X\mapsto X+\xi_l,Y\mapsto Y)$. The fact that
      $\nicefrac{\Kbb_l(X)[Y]}{N_i}\rightarrow \Kbb_i((T))$, $X\mapsto
      \tilde{x_i}$, $Y\mapsto\tilde{Y_i}$ is dense is just
      Proposition~\ref{prop_ratpuiseux}. It follows that $F\rightarrow
      \Kbb_i((T))$ is dense and since $F/K_N$ is a constant field extension
      it preserves completions which proves the density. Since the
      valuation of $X-\xi$ is $e_i$ this shows that the associated place
      lies above $l$.\\

      To show that it is surjective, let $\Pfrak$ be a place above
      $l$, and $\Pfrak'\in\Pbb_F$ the unique place lying above $\Pfrak$ and
      $X-\xi$. Through the shift $X\mapsto X+\xi_l$ we may assume that
      $\Pfrak'$ lies above $X$. Then according to
      Corollary~\ref{cor_place-factor}, there is an irreducible factor
      $N_i$ of
      $N_*$ in $\Kbb_l\ldbrack X\rdbrack[Y]$ such that $F_{\Pfrak'}\simeq
      \Kbb_l((x))[Y]/N_i$. According to Proposition~\ref{prop_ratpuiseux},
      if $\tilde{x},\tilde{y}$ is a rational Puiseux expansion of $N_*$
      associated to the factor $N_i$, of coefficient field $\Kbb$ then the
      morphism $(F\rightarrow \Kbb((T))$, $X\mapsto \tilde{x}$, $Y\mapsto
      \tilde{y}$) is continuous and dense. Thus we
      have $K_{N_*,\Pfrak}\simeq F_{\Pfrak'}$ which concludes the proof.
    \end{proof}
    Lemma~\ref{lemme_proofirredtesthm} proves that steps (3) let us compute
    parametrisation of $K_{N,\Pfrak}$ for all the finite poles of $a$ (in
    fact a little more because we also consider the integral Puiseux
    expansions of $N$ but those are cancelled by the condition
    $1-e_i-\nu(\lceil y_i\rceil)\geq 0$ and do not change the overall
    complexity) and
    that all the computations that we do correspond to one of those. If
    that was not the case then we might have a test failing that correspond
    to none of the places of $K_N$.\\

    In a similar manner, we show that step $(4)$ $(5)$ capture all the
    places at infinity. The proof of the
    correctness of Algorithm~\ref{algo_irreducibility_rational} now follows
    easily from Lemma~\ref{lemma_simplesyslocal}\\

    For the proof of complexity we first evaluate the cost of steps (3a) to
    (3d) in terms of arithmetic operations in $\Kbb_l$. Step (3b) is a
    shift of $d_y+1$ polynomials of degree at most $d_x$ in $\Kbb_l[X]$ and can thus be done
    $\tilde{O}(d_yd_x)$ arithmetic operations in $\Kbb_l$. We state in
    Theorem~\ref{theorem_compratPuiseux}, following~\cite{PoWe21} that
    step (3c) can be done in $\tilde{O(d_y\delta)}$ operations in $\Kbb_l$
    where $\delta$ is the valuation in $X$ of
    $\mathrm{Res}_Y(N_*(X+\xi_l,Y),\partial_YN_*(X+\xi_l,Y))$. This is the same as the
    valuation in $l$ of $\mathrm{Res}_Y(N_*,\partial_YN_*)$ which we denote
    $\delta_l$.
    We finally study the cost of step (d). For each $i\in\ldbrack
    1;g_l\rdbrack$, we denote $f_i=[\Kbb_i:\Kbb_l]$. Note that we must have
    $\sum_{i=1}^{g_l}e_if_i\leq d_y$. Indeed, $N_*$ factors in over
    $\Kbb_l[X,Y]$ as a product of conjugated irreducible factors of the
    same degree $d$. If we group our rational Puiseux expansion depending on
    which factor of $N_*$ they come from, we have $\sum e_if_i\leq m$ for
    each group. As they are as many groups as they are factors of $N_*$ we
    have the aforementioned inequality. Computing the rational Puiseux
    expansions from their singular part up to the desired precision can be
    done in quasilinear time following~\cite[Corollary~5.1 and 5.2 p.
    251]{KuTr78} so $\tilde{O}(\pm e_if_i+f_i\nu_i((a)_-))$ operations in
    $\Kbb_l$ (the $\pm e_if_i$ part is in fact negative for finite places so it
    only matters for the places at infinity.

    In total, step (3) excluding solving the linear system in (3b) cost
    \[\tilde{O}(\sum_{l|l_c}\deg(l)(d_y
    d_x+d_y\delta_l+\sum_{\Pfrak|l}e_\Pfrak[\Gcal_\Pfrak:\Kbb_l]+[\Gcal_\Pfrak:\Kbb_l]\nu_\Pfrak((a)_-)))\]
    arithmetic operations over $\Fbb_p$. Now using that
    $\sum_{l|l_c}\delta_l\deg(l)\leq\deg(\mathrm{Res}_Y(N_*,\partial_YN_*))=\tilde{O}(d_yd_x)$,
    $\sum_{\Pfrak|l}e_\Pfrak[\Gcal_\Pfrak:\Kbb_l]\leq d_y$ and
    $\sum_{\Pfrak}[\Gcal_\Pfrak:\Kbb_l]\deg(l)\nu_\Pfrak((a)_-)=\deg((a)_-)=d_x$
    we find that all of those computations can be done in
    $\tilde{O}(d_yd_x^2+d_y^2d_x)$ arithmetic operations in $\Fbb_q$.

    The step $(3b)$ can be done in
    $\tilde{O}((b^2f_i^2\deg(l)^2\log(p)+b^\omega f_i^\omega\deg(l)^{\omega}))$
    arithmetic operations in $\Fbb_p$ (where $\omega$ is an exponent for
    the matrix multiplication). We remember that this computation
    can only happen for the place which are poles of $(x)$ or poles of
    $(a)$. Thus the total cost is at most
    $\tilde{O}(b^2(\deg(x)_-+\deg(a)_-)^2\log(p)+b^\omega(\deg(x)_-+\deg(a)_-)^\omega)=\tilde{O}(b^2(d_y+d_x)^2\log(p)+
    b^\omega(d_y+d_x)^\omega)$ arithmetic operations in $\Fbb_p$.\\

    $P_\infty$ has the same overall degrees as $N_*$ so the cost of step
    $(4)$ does not change the overall complexity.

    It follows that the overall bit complexity of
    Algorithm~\ref{algo_irreducibility_rational} is
    $\tilde{O}((d_yd_x^2+d_y^2dx+b^{\omega-1}(d_x+d_y)^\omega)b\log(p)+b^2(d_x+d_y)^2\log(p)^2)$
    bit operations.

  \end{proof}

\section{Solving the norm equation equation}\label{section_resolution}
    The goal of this section is to present an algorithm to solve the
    $p$-Riccati equation relative to an irreducible polynomial $N\in C[Y]$
    over $K_N$. We discuss its complexity and its applications to the
    factorisation of differential operators in $K\langle\partial\rangle$.
    This algorithm makes use of algebraic
    geometry tools such as Riemann-Roch spaces and the Picard group of $K_N$.\\
    We recall that van der Put and Singer presented,
    in~\cite[§13.2.1]{PuSi03}, a method to compute
    solutions of $p$-Riccati equations over $\Fbb_q(x)$. Their method will serve as a guideline
    for the techniques we develop in the general case.\\
    We keep the notations of the previous section. In addition, we suppose
    that $N\in C[Y]$ is a fixed irreducible polynomial and use the
    notations introduced in Notation~\ref{notation_K_N}. In particular,
    we recall that $S_N$ denotes the set of solutions of the $p$-Riccati
    equation $f^{(p-1)}+f^p=y_N$.
  \begin{Proposition}\label{borne_val}
    Let $\Pfrak$ be a place of $K_N$, $t_{\Pfrak}$ be a prime element of
    $\Pfrak$. Then, for all $f\in S_N$, we have
    $\nu_\Pfrak(f)\geq\min(p^{-1}\nu_{\Pfrak}(y_N),\nu_\Pfrak(t'_\Pfrak)-1)$.
  \end{Proposition}
  \begin{proof}
    We have
    \begin{align*}
      \nu_{\Pfrak}(y_N)&=\nu_{\Pfrak}(f^{(p-1)}+f^p)\\
      &\geq\min(\nu_{\Pfrak}(f^{(p-1)},p\nu_{\Pfrak}(f)).
    \end{align*}
    Furthermore equality holds if
    $\nu_{\Pfrak}(f^{(p-1)})\neq p\nu_{\Pfrak}(f)$.
    Since $\nu_\Pfrak(f^{(p-1)})\geq
    \nu_\Pfrak(f)+(p-1)(\nu_\Pfrak(t_\Pfrak')-1)$ if
    $\nu_{\Pfrak}(f)<\nu_{\Pfrak}(t_\Pfrak')-1$, we find in particular that
    $p\nu_{\Pfrak}(f)<\nu_{\Pfrak}(f)+(p-1)(\nu_{\Pfrak}(t_\Pfrak')-1)\leq
    \nu_{\Pfrak}(f^{(p-1)})$ so $\nu_{\Pfrak}(y_N)=p\nu_{\Pfrak}(f)$.
  \end{proof}
  In fact we can show that if solutions exists, some of them verify a
  slightly better bound.
  \begin{Definition-Proposition}\label{residue_manip}
    Let $f\in S_N$ and $\Pfrak$ be a
    place of $K_N$ verifying
    $\nu_{\Pfrak}(y_N)\geq p\nu_\Pfrak(t_\Pfrak')$. Then
    there exists $k\in\Fbb_p$ such that for all $g\in K_N$, if
    $\nu_{\Pfrak}(g)\equiv k\pmod p$ then $f-\frac{g'}{g}\in S_N$ and
    $\nu_{\Pfrak}\left(f-\frac{g'}{g}\right)\geq \nu_\Pfrak(t_\Pfrak')$.
    We call $k$ the ramified residue of $f$ in $\Pfrak$ and write 
    \[\mathfrak{Re}_{\Pfrak}(f):=k.\]
  \end{Definition-Proposition}
  \begin{proof}
  
    If $\nu_{\Pfrak}(f)\geq \nu_\Pfrak(t_\Pfrak')$ then we can take $k=0$. Indeed in this case,
    if $\nu_{\Pfrak}(g)\equiv 0\pmod p$ then there exists $l\in\Nbb$ such
    that $g=t_{\Pfrak}^{pl}u$ with $\nu_{\Pfrak}(u)=0$. Then 
    $\frac{g'}{g}=\frac{u'}{u}+pl\frac{t'_{\Pfrak}}{t_{\Pfrak}}=\frac{u'}{u}$.
    Since $\nu_\Pfrak(u)=0$, we can write
    $u=\sum_{n=0}^{\infty}u_nt_\Pfrak^n$ and
    $u'=t_\Pfrak'\sum_{k=0}^{\infty}(n+1)u_{n+1}t_\Pfrak^n$. Thus
    $\nu_\Pfrak(u')\geq \nu_\Pfrak(t_\Pfrak')$ and
    $\nu_{\Pfrak}\left(\frac{g'}{g}\right)\geq \nu_\Pfrak(t_\Pfrak')$ which yields the
    result.\\
    Suppose now that $\nu_{\Pfrak}(f)=\nu_\Pfrak(t_\Pfrak')-1$. We set
    $e=1-\nu_\Pfrak(t_\Pfrak')$, $a:=(t_\Pfrak^{e-1}t_\Pfrak')(\Pfrak)$ and 
    $c:=(t_\Pfrak^ef)(\Pfrak)$. Let us show
    that $c\in\Fbb_p^\times a$. The characteristic $p$ does not divide $e$,
    and we know that $\nu_\Pfrak(f^{(p-1)})=-pe$. Furthermore we know 
    (Lemma~\ref{p-1_power_derivation}) that
    $f^{(p-1)}:=\frac{\ud^{p-1}}{\ud t_\Pfrak^{p-1}}(t_\Pfrak^{p-1}f)$.
    It follows that
    $(t_\Pfrak^{pe}f^{(p-1)})(\Pfrak)=-a^{p-1}c$ 
    and 
    $(t_\Pfrak^{pe}f^p)(\Pfrak)=c^p$.
    But 
    $t_\Pfrak^{pe}(f^{(p-1)}+f^p)(\Pfrak)=(t_{\Pfrak}^{pe}y_N)(\Pfrak)=0$
    since $\nu_\Pfrak(y_N)>-pe$.
    It follows that
    $t_\Pfrak^{pe}(f^{(p-1)}+f^p)(\Pfrak)=c^p-a^{p-1}c=0$. Thus $c^{p-1}=a^{p-1}$ and $c\in\Fbb_p^{\times}a$. We set $k:=c\cdot
    a^{-1}$.

    Let $g\in K_N$ be such that $\nu_\Pfrak(g)\equiv k\pmod p$.
    There exists $l\in\Zbb$ and $u\in K_N$ such that $\nu_\Pfrak(u)=0$ and
    $g=t_\Pfrak^{pl+k}u$. Then
    $\frac{g'}{g}=k\frac{t_\Pfrak}{t_\Pfrak}+\frac{u'}{u}$.
    Since $\nu_\Pfrak(u)=0$, $\nu_\Pfrak(u')> -e$ and
    $\nu_\Pfrak\left(\frac{g'}{g}\right)=-e$. Then
    $\left(t_\Pfrak^{e}\frac{g'}{g}\right)(\Pfrak)=k(t_\Pfrak^{e-1}t_\Pfrak)(\Pfrak)=ka=c$.
    Thus
    $\left(t_\Pfrak^e\left(f-\frac{g'}{g}\right)\right)(\Pfrak)=0$,
    which is to say that
    $\nu_\Pfrak\left(f-\frac{g'}{g}\right)\geq 1-e=\nu_\Pfrak(t_\Pfrak')$.
  \end{proof}
   In particular if $\Pfrak$ is neither ramified nor a
  pole of $y_N$ then $S_N$ contains an element with no pole in $\Pfrak$.
  This local improvement on the bound provided in
  Proposition~\ref{borne_val} is accomplished by adding an element of the form
  $\frac{g'}{g}$. Unfortunately adding such an element makes new poles
  appear in general so this local approach is not enough.
  We globalize it in the following theorem.
  \begin{Theoreme}\label{theo_local_sol}
    Let $f\in S_N$ and $S:=\{\Pfrak\in \Pbb_{K_N}|\nu_{\Pfrak}(y_N)<
    p\nu_\Pfrak(t_\Pfrak')\}$. Set
    \[\mathfrak{Re}(f):=\sum_{\substack{\Pfrak\in\Pbb_{K_N}\\\Pfrak\notin
    S}}\mathfrak{Re}_{\Pfrak}(f)\cdot\Pfrak.\]
    If there exist $D',D_p\in\mathrm{Div}(F)$ such that 
    $\mathfrak{Re}(f)\sim D'+pD_P$
    then $S_N$ contains an element $\varphi$ verifying for all places $\Pfrak$ outside $S\cup
    \mathrm{supp}(D')$ that 
    $\nu_{\Pfrak}(\varphi)\geq \nu_\Pfrak(t_\Pfrak')$.
  \end{Theoreme}
  \begin{proof}
    Since $\mathfrak{Re}(f)\sim D'+pD_p$, there exists $g\in K_N$ such that $\mathfrak{Re}(f)-D'-pD_p=(g)$.
    From Lemma~\ref{affine_S_N}, we deduce that $f-\frac{g'}{g}\in S_N$.
    Let $\Pfrak\in\Pbb_{K_N}\backslash(S\cup \mathrm{supp}(D'))$.
    Then we find
        \begin{align*}
          \nu_{\Pfrak}(g)&=\nu_{\Pfrak}(\mathfrak{Re}(f))-\nu_{\Pfrak}(D')-p\nu_{\Pfrak}(D_p)\\
          &=\nu_{\Pfrak}(\mathfrak{Re}(f))-0-p\nu_{\Pfrak}(D_p)\\
          &\equiv \nu_\Pfrak(\mathfrak{Re}(f))\pmod p\\
          &\equiv \mathfrak{Re}_{\Pfrak}(f)\pmod p
        \end{align*}
        By definition of $\mathfrak{Re}_{\Pfrak}(f)$, $f-\frac{g'}{g}$ is an
        element of $S_N$ verifying for any place $\Pfrak$ outside
        $S\cup\mathrm{Supp}(D')$ that
        \[\nu_\Pfrak(f-\frac{g'}{g})\geq \nu_\Pfrak(t_\Pfrak').\]
  \end{proof}
  \begin{Definition}
    We consider
    $\Gfrak^p_N=\text{\large$\nicefrac{\mathrm{Cl}(K_N)}{p\mathrm{Cl}(K_N)}$}$.
    Since $K_N$ is an algebraic function field of characteristic $p$, $\Gfrak^p_N$ is a 
    finite commutative group of the form $\Gfrak^p_N\simeq
    \left(\text{\Large$\nicefrac{\Zbb}{p\Zbb}$}\right)^n$
    for some $n\in\Nbb^*$.
  \end{Definition}
  \begin{Corollaire}\label{local_sol}
    For each place $\Pfrak\in\Pbb_{K_N}$ we denote by $t_\Pfrak$ a prime
    element of $\Pfrak$.\\
    Let $(D_1,\ldots,D_n)\in\mathrm{Div}(K_N)^n$ be a lifting of a generating
    family of
    $\Gfrak^p_N$ viewed as a $\Fbb_p$-vector space.\\
    Let $S=\bigcup_{i=1}^n\Supp D_N$ and set
    \[A:=\max\left(\sum_{\Pfrak\in
    S}\Pfrak+\mathrm{Diff}(K_N)-2(x)_-,\frac{(y_N)_-}{p}\right).\]
    If $S_N$ is not empty then it contains an element of $\Lcal(A)$.
  \end{Corollaire}
  \begin{proof}
    Let $f\in S_N$ and let $\mathfrak{Re}(f)$ be defined similarly as in
    Theorem~\ref{theo_local_sol}. Since $D_1,\ldots,D_n$ is a basis of
    $\Gfrak^p_N$, there exists a linear combination $D'=a_1D_1+\ldots+a_nD_n$ such
    that $\mathfrak{Re}(f)\equiv D'\pmod{\Gfrak_N^p}$. Thus there exists
    $D_p\in\mathrm{Div}(K_N)$ such that
    \[\mathfrak{Re}(f)\sim D'+pD_p.\]
    Besides $\mathrm{supp}(D')\subset\bigcup_{i=1}^n\mathrm{supp}(D_i)\subset S$.
    According to Theorem~\ref{theo_local_sol}, $S_N$ contains an element
    $f^*$ verifying for all places outside of $S$ that $\nu_{\Pfrak}(f^*)\geq
    \nu_\Pfrak(t_\Pfrak')$.
    The corollary now follows from Proposition~\ref{borne_val} and the fact that
    the valuation in $\Pfrak$ of the divisor $\Diff(K_N)-2(x)_-$ is
    precisely $-\nu_\Pfrak(t_\Pfrak')$.
  \end{proof}
  \begin{Definition}\label{construct_divisor}
    For any effective divisor $D$ over $K_N$, we define
    \[A(D):=\max\left(\sum_{\Pfrak\in\Supp D}\Pfrak
    +\mathrm{Diff}(K_N)-2(x)_-,\frac{(y_N)_-}{p}\right).\]
    We say that $D$ is a generating divisor of $\Gfrak^p_N$ if and only if
    $(\Pfrak)_{\Pfrak\in\Supp D}$ is a generating family of $\Gfrak^p_N$.
    In this case 
    \[S_{N}=\varnothing\Leftrightarrow S_N\cap\Lcal(A(D))=\varnothing.\]

    For a family of effective divisors $(D_1,\ldots,D_n)$ we define
    \[A(D_1,\ldots,D_n)=A(D_1+\ldots+ D_n).\]
  \end{Definition}

  To our knowledge there exists no algorithm able to compute the cokernel
  of the multiplication by $p$ in the divisor class group of a curve
  $\Ccal$ of genus $g$ in polynomial time in $g$ and the characteristic
  $p$. We instead opt to choose enough uniformly random elements of
  $\Gfrak^p_N$ to generate the whole group. Since we know that $\Gfrak^p_N$
  is of the form $\Fbb_p^n$ with $n$ being an integer smaller than $g+1$,
  we know that we only need to select $O(g)$ elements on average. We use  
  Algorithm~\ref{algo_irreducibility} beforehand to ensure the existence of
  a solution. We refer to~\cite[Section~3.5]{Bruin10} in which the author
  present an algorithm to select uniformly random elements in
  $\mathrm{Cl}^0(K_N)$. If $K_N$ is seen as the regular function field of a
  curve $\Ccal$, \cite[Algorithm~3.7]{Bruin10} presupposes the choice of a
  line bundle $\Lcal$ over $\Ccal$ of degree at least $2g+1$. Since we are
  working over finite fields, line bundles of arbitrary degrees exists and
  we can choose a line bundle of degree exactly $2g+1$. Then we can use 
  \cite[Algorithm~3.7]{Bruin10} to pick uniformly random elements in
  $\mathrm{Cl}^0(K_N)$ represented by uniformly random effective divisors of degree 
  $2g+1$ in polynomial time in $g$ and $\log(q)$. 
  However, \cite[Algorithm~3.7]{Bruin10} also suppose that the zeta
  function of $\Ccal$ is known in order to ensure the uniform distribution of the divisors.
  The computation of the zeta function can be done in time polynomial in
  $g$ and linear in $b$ and $p$ (precisely $\tilde{O}(pbd_x^6d_y^4)$ bit
  operations~\cite{tuitman17}). 
  \begin{Remarque}
    In~\cite[section~13.2]{BCEJM10} the authors also state that
    $\mathrm{Cl}^0(K_N)$ is generated by the places of degree less than
    $1+2\log_q(4g-2)$. This in turns guarantees the existence of $D$ a
    generating divisor of $\Gfrak_{N_*}^p$ of degree $\tilde{O}(d_xd_y)$. 
    However the probability of generating
    $\Gfrak^p_N$ with $O(g)$ uniformly chosen effective divisors of degree
    less than $1+\log_q(4g-2)$ could be very low which is why we do not use
    it for our algorithm.
  \end{Remarque}
  From now on we will assume that we are able to pick uniformly random
  elements in $\mathrm{Cl}^0(K_N)$. If $\Gfrak^{p}_N$ is of dimension $r$ over
  $\Fbb_p$ then we only need on average to select $O(r)$ elements to generate $\Gfrak^{p}_N$.\\

  We now discuss in more details the computation of the linear system
  representing the $p$-Riccati equation over some
  $\Lcal(A(D))$. The main issue lies in the computation of the $(p-1)$-th
  derivative of the elements of a basis of $\Lcal(A(D))$. Instead of
  computing their exact value in $K_N$, we compute their Taylor expansion
  up to a high enough precision.

  \begin{Proposition}\label{Proposition_section_stability}
    Let $\Phi_N:K_N\rightarrow C_N$ denote the Frobenius endomorphism over
    $K_N$ and $D\in\mathrm{Div}(K_N)$. Let $f\in\Lcal(A(D))$. Then
    $\Phi_N^{-1}(f^{(p-1)}+f^p)\in\Lcal(A(D))$.
  \end{Proposition}
  \begin{proof}
    Let $\Pfrak\in\Pbb_{K_N}$. If $\Pfrak\notin\Supp(A(D))$ then by definition
    of $A(D)$, $f$ is not a ramified place and $f$ has no poles in $\Pfrak$.
    Thus neither $f^{(p-1)}$ nor $f^p$ has a pole in $\Pfrak$. Thus
    $\Phi_N^{-1}(f^{(p-1)}+f^p)$ has no pole in $\Pfrak$.
    For $\Pfrak\in\Supp(A(D))$, we let $t_\Pfrak$ be a prime element of
    $\Pfrak$. We know that 
    \[
      \nu_\Pfrak(\Phi_N^{-1}(f^{(p-1)}+f^p))\geq \min (p^{-1}\cdot
      (\nu_\Pfrak(f)+(p-1)(\nu_\Pfrak(t_\Pfrak')-1)),\nu_\Pfrak(f))
    \]
    Besides we know that if $\nu_\Pfrak(f)\leq\nu_\Pfrak(t_\Pfrak')-1$ then
    $p^{-1}\cdot(\nu_\Pfrak(f)+(p-1)(\nu_\Pfrak(t_\Pfrak')-1))\geq\nu_\Pfrak(f)$
    so in that case we get that 
    $\nu_\Pfrak(\Phi_N^{-1}(f^{(p-1)}+f^p))\geq \nu_\Pfrak(f)$ which
    implies the desired result since $f\in\Lcal(A(D))$.
    If now we have $\nu_\Pfrak(f)>\nu_\Pfrak(t_\Pfrak')-1$ then 
    $p^{-1}\cdot(\nu_\Pfrak(f)+(p-1)(\nu_\Pfrak(t_\Pfrak')-1))>\nu_\Pfrak(t'_\Pfrak)-1$.
    Since valuations have to be integers we deduce that
    $\nu_\Pfrak(\Phi_N^{-1}(f^{(p-1)}+f^p))\geq \nu_\Pfrak(t'_\Pfrak)\geq
    -\nu_\Pfrak(A(D))$. Thus $\Phi_N^{-1}(f^{(p-1)}+f^p)\in\Lcal(A(D))$.
  \end{proof}
  \begin{Notation}
    From this point forward we assume that $K=\Fbb_{p^b}(x)$, for
    $b\in\Nbb^*$. We suppose that
    $N\in \Fbb_{p^b}[x^p,y]$ is a fixed separable irreducible polynomial and set
    $N_*(x,y)\in \Fbb_{p^b}[x,y]$ such that $N_*^p(Y)=N(Y^p)$. Let $a$ be the $p$-th
    root of $y_N\in K_N$. Then $N_*(a)=0$ and $K_N=K[a]$. We set
    $d_x=\deg_x N_*$ and $d_y=\deg_Y N_*$.
  \end{Notation}
  \begin{Definition}
    Let $f\in K_N$. There exist unique $f_0,\ldots, f_{p-1}\in K_N$ such
    that
    \[f=\sum_{i=0}^{p-1}f_i^{p}x^i.\]
    For all $i\in\ldbrack 0;p-1\rdbrack$ We denote by $S_i(f):= f_i$ the
    $i$-th section of $f$.
  \end{Definition}
  Although we define sections for all $i\in\ldbrack0;p-1\rdbrack$, we will
  really only be interested in $S_{p-1}$ as shown in the following lemma:
  \begin{Lemme}\label{section_derivation}
    For any $f\in K_N$, \[\Phi_N^{-1}(f^{(p-1)})=-S_{p-1}(f).\]
  \end{Lemme}
  \begin{proof}
    Let $f:=\sum_{i=0}^{p-1}f_i^p x^i$. It suffices to show that
    $f^{(p-1)}=-f_{p-1}^p$. But this is obvious since
    $f^{(p-1)}=(p-1)!f_{p-1}^p$ and $(p-1)!=-1\mod p$.
  \end{proof}
  Thus another way of writing $p$-Riccati equation is
  \[b-S_{p-1}(b)=a.\]
  We now use the fact that Lemma~\ref{section_derivation} also holds over $K_{N,\Pfrak}$ 
  for any $\Pfrak\in\Pbb_{K_N}$. Let $\Pfrak$
  be a place over $K_N$ that does not belong in $\Supp(A(D))$. Then the
  injective homomorphism from $K_N$ to its $\Pfrak$-completion induces an injective 
  homomorphism of $\Fbb_q$-vector spaces
  $\Lcal(A(D))\hookrightarrow \Gcal_\Pfrak[[t_\Pfrak]]$.
  It follows that there exists a constant $N\in\Nbb$ such that for all
  $f\in\Lcal(A(D))$, $f=0$ if and only if $\nu_\Pfrak(f)\geq N$.
  \begin{Lemme}
    Let $\Pfrak\in\mathrm{Div}(K_N)$, let $D$ be an effective divisor of $K_N$ and
    set $d:=\deg(A(D))$. For any $f\in\Lcal(A(D))$,
    \[f=0\Leftrightarrow \nu_\Pfrak(f)>\frac{d}{\deg(\Pfrak)}.\]
  \end{Lemme}
  \begin{proof}
    Since $f\in\Lcal(A(D))$, if $f\neq 0$ then we know that $\deg
    (f)_-\leq d$. But since $\deg
    (f)_-=\deg (f)_0$ we know that $\nu_\Pfrak(f)\leq
    \frac{\deg(f)_-}{\deg(\Pfrak)}\leq \frac{d}{\deg(\Pfrak)}$.
  \end{proof}
  Thus it suffices for a function $f\in\Bcal$ (where $\Bcal$ is a basis of
  $\Lcal(A(D))$) to compute the image of $f-S_{p-1}(f)$ modulo
  $t_\Pfrak^{\left\lfloor \frac{\deg(A(D))}{\deg \Pfrak}\right\rfloor+1}$ in
  $\Gcal_{\Pfrak}[[t_\Pfrak]]$. If one writes $f=\sum_{k=0}^{\infty} f_i
  t_\Pfrak^i$ then $S_{p-1}(f)\mod t_\Pfrak^{\left\lfloor
  \frac{\deg(A(D))}{\deg \Pfrak}\right\rfloor+1}$ can be deduced from the knowledge of the
  coefficients $f_{pk+p-1}$ for $k\leq \frac{deg(A(D))}{\deg\Pfrak}$.
  To that end we can compute the first $p\left\lfloor \frac{\deg(A(D))}
  {\deg\Pfrak}\right\rfloor+p-1$ coefficients of the Taylor expansion of
  $f$. In practice, we compute the Taylor expansion of $a$ of which we know the minimal
  polynomial, in
  $t_\Pfrak$ up to the desired precision by Newton iteration (note that by
  definition of $A(D)$, $a\in\Lcal(A(D))$). This can be done in $\tilde{O}(p\deg(A(D))d_y)$
  operations in $\Fbb_q$. Then, knowing that elements of $\Lcal(A(D))$ are given by
  polynomials $F(x,a)$ we get their Taylor expansions by composition for an
  additional cost of $\tilde{O}(p\deg(A(D))d_y)$ operations in $\Fbb_q$.
\begin{Proposition}\label{coefficient_trace_formula}
  Let $Q_i$ be the quotient of the Euclidean division of $N_*(x,y)$ by $y^{i+1}$ for
  any $i\in\Nbb$. Then for any $f:=\sum_{k=0}^{d_y-1} f_ka^k\in K_{N}$
  and any $i\in\ldbrack 0;d_y-1\rdbrack$,
  $f_i=\Tr_{K_{N}/\Fbb_q(x)}\left(\frac{Q_i(x,a) f}{\partial_y
  N_*(x,a)}\right)$.
\end{Proposition}
\begin{proof}
  Let us fix $N_*(x,y)=\sum_{k=0}^{d_y} \eta_k(x) y^k$. From \cite[Lemma~2 section~III.
  6]{Serre04} we know that $\Tr_{K_{N}/\Fbb_{p^b}(x)}\left(\frac{a^{i}}{\partial_y
  N_*(x,a)}\right)=\frac{1}{\eta_{d_y}}\delta_{i,d_y-1}$,  for all
  $i\leq d-1$. Thus the result holds for $i=d_y-1$, since $Q_{d_y-1}=\eta_{d_y}$. Then for all
  $i$ we have $Q_i=Q_{i+1} y+\eta_{i+1}$. We assume the proposition to be
  true for $i+1$. Then 
  \[
    \Tr_{K_{N}/\Fbb_{p^b}(x)}\left(\frac{Q_i(x,a) f}{\partial_y
    N_*(x,a)}\right)=\Tr_{K_{N}/\Fbb_{p^b}(x)}\left(\frac{Q_{i+1}(x,a)a f}{\partial_y
    N_*(x,a)}\right)+\eta_{i+1}\Tr_{K_{N}/\Fbb_{p^b}(x)}\left(\frac{f}{\partial_y
    N_*(x,a)}\right)\]
  and by hypothesis $\Tr_{K_{N}/\Fbb_{p^b}(x)}\left(\frac{Q_{i+1}(x,a)a f}{\partial_y
    N_*(x,a)}\right)$ is the coefficient of $a^{i+1}$ in $af$, which is
    given by $f_i-\frac{f_{d_y-1}\eta_{i+1}}{\eta_{d_y}}$, while $Tr_{K_{N}
    /\Fbb_{p^b}(x)}\left(\frac{f}{\partial_y N_*(x,a)}\right)$ is the
    coefficient of $a^{d_y-1}$ of $\frac{f}{\eta_{d_y}}$.\\

    \[\Tr_{K_{N}/\Fbb_{p^b}(x)}\left(\frac{Q_i(x,a) f}{\partial_y
    N_*(x,a)}\right)=f_i-\frac{f_{d_y-1}\eta_{i+1}}{\eta_{d_y}}+\eta_{i+1}\frac{f_{d_y-
    1}}{\eta_{d_y}}=f_i.\]
\end{proof}
\begin{Corollaire}\label{coeff_poles_trace}  
  Let $D$ be an effective divisor over $K_N$ and $P$ be an irreducible
  polynomial in $\Fbb_{p^b}[x]$ coprime with
  $\mathrm{Disc}(N_*)$ and the leading coefficient of $N_*$. If none of the
  places in $\Supp(D)$ lie above $P$ then for any
  $f\in\Lcal(A(D))$, none of the coefficients of $f$ in the basis
  $(1,a,\ldots,a^{d_y-1})$ have a pole in $P$.
\end{Corollaire}
\begin{proof}
  Let $l_c$ be the leading coefficient of $N_*$. The function $l_c a$ is integral and
  its minimal polynomial is $N_1=l_c^{d_y-1} N_*(x,Y/l_c)$. We have
  $\Disc(N_1)=l_c^{d_y-1}\Disc(N_*)$ and $(\partial_Y N_*(x,a))_+\leq
  (\partial_Y N_1(x,l_c a))_0
  \leq (d_y-1)(l_c)_0+(\mathrm{Disc}(N_*))_0$. This shows that for all $i$
  (with the notations of the previous proposition),
  $\frac{Q_i(x,a)}{\partial_yN_*(x,a)}$ has no poles that divides $P$ since
  the poles of $a$ are among those of $l_c$.\\
  Let us now show that $\Supp(A(D))$
  does not contain any place that divides $P$. 
  Let $\mathrm{Diff}(K_{N_*})_0$ be the different divisor of $K_{N_*}$ outside of
  the places at infinity. Since $l_c a$ is integral we know that
  $\mathrm{Diff}(K_{N_*})_0\leq (\partial_Y N_1(x,l_c a))_0\leq
  (d_y-1)(l_c)_0+(\mathrm{Disc}(N_*))_0$. Thus $\Supp(A(D))\cap\Supp(P)\subset
  (\Supp(\Diff(K_N))\cap\Supp (P))\cup
  (\Supp(a)_-\cap\Supp(P)\subset\Supp(l_c)\cap\Supp(P)=\varnothing$.
  Thus if we set $\Ocal_P$ the valuation ring associated to $P$ in
  $\Fbb_{p^b}(x)$ and
  $\Ocal_P'$ its integral closure in $K_{N}$, then for all $i$ and all
  $f\in \Lcal(A(D))$,  $\frac{Q_i(x,a)f}{\partial_yN_*(x,a)}\in \Ocal_P'$.
  It follows that if $f_i$ denotes the $i$-th coefficient of $f$ then
  $f_i=\Tr_{K_N/\Fbb_{p^b}(x)}\left(\frac{Q_i(x,a)f}{\partial_yN_*(x,a)}\right)\in
  O_P$ and $f_i$ has no pole in $P$.
\end{proof} 
When knowing the Taylor expansion of $a$ up to the desired precision, computing
the Taylor expansion of an element $f$ of $\Lcal(A(D))$ by composition requires to
compute the Taylor expansion of its coefficients. This can be done in
$\tilde{O}(p\max(\eta,\deg A(D)) d_y)$ operations in $\Fbb_{p^b}$ where $\eta$ is
the degree of the coefficients of $f$. As we show now, by
construction of $A(D)$, $\eta$ and $\deg(A(D))$ have the same order of magnitude.
  \begin{Lemme}
    Let $f\in K_N$ and $\Pfrak\in\Pbb_{\Fbb_{p^b}(x)}$. 
  \[\nu_{\Pfrak}(\Tr_{K_N/\Fbb_{p^b}(x)}(f))
    \geq \min_{\Pfrak'|\Pfrak}\left
    \lfloor\frac{\nu_{\Pfrak'}(f)}{e(\Pfrak'|\Pfrak)}\right\rfloor.\]
  \end{Lemme}
\begin{proof}
  Let $\Ocal_\Pfrak$ be the valuation ring associated to the place
  $\Pfrak$ and $\Ocal_\Pfrak'$ be its integral closure in $K_N$.
  For any $f\in K_N$, if $f\in \Ocal_\Pfrak'$ then $\Tr_{K_N/\Fbb_{p^b}(x)}(f)\in\Ocal_{\Pfrak}$
  \cite[Corollary~3.3.2]{StHe08}.\\
  It follows that if $\Pfrak$ is a pole of $\Tr_{K_N/\Fbb_{p^b}(x)}(f)$, then at least
  one of the places lying under $\Pfrak$ is a pole of $f$. Let
  $\Pfrak^*$ above $\Pfrak$ be such that
  \[\left\lfloor
  \frac{\nu_{\Pfrak^*}(f)}{e(\Pfrak^*|\Pfrak)}\right\rfloor=\min_{\Pfrak'|\Pfrak}\left                  \lfloor\frac{\nu_{\Pfrak'}(f)}{e(\Pfrak'|\Pfrak)}\right\rfloor.\]
  Set $k=\left\lceil
  \frac{-\nu_{\Pfrak^*}(f)}{e(\Pfrak^*|\Pfrak)}\right\rceil$ and  $P\in K_N$
  a prime element of $\Pfrak$. Then for any $\Pfrak'$ above $\Pfrak$ we have

  \[
    \nu_{\Pfrak'}(P^k f)=ke(\Pfrak'|\Pfrak)+\nu_{\Pfrak'}(f).\]
  By definition $k\geq -\frac{\nu_{\Pfrak'}(f)}{e(\Pfrak'|\Pfrak)}$ thus
  $\nu_{\Pfrak'}(P^kf)\geq 0$. It follows that
  \begin{align*}
    \nu_{\Pfrak}(\Tr_{K_N/\Fbb_{p^b}(x)}(P^k f))&=
    \nu_{\Pfrak}(P^k\Tr_{K_N/\Fbb_{p^b}(x)}(f))\\
    &=k+\nu_{\Pfrak}(\Tr_{K_N/\Fbb_{p^b}(x)}(f))\\
    &\geq 0\\
    \nu_{\Pfrak}(\Tr_{K_N/\Fbb_{p^b}(x)}(f))&\geq -k
  \end{align*}
  which is the desired result.
\end{proof}
\begin{Proposition}\label{divisor<->coeff}
  Let $D$ be an effective divisor over $K_N$ and $f=\frac{1}{f_{-1}}\sum_{i=0}^{d_y-1}f_i
  a^i\in\Lcal(A(D))$ where $f_{-1},f_0,\ldots,f_{d_y-1}\in\Fbb_q[x]$ are
  globally coprime polynomials.\\
  Then for any $i\in\ldbrack-1;d_y-1\rdbrack$, both $\deg(f_i)$ and
  $\deg(A(D))$ are in $O(\deg(D)+d_xd_y)$.\\
\end{Proposition}
\begin{proof}

  Let $P\in\Fbb_q[x]$ be an irreducible polynomial and let $Q_i$ denote the
  quotient of the Euclidean division of $N_*$ by $Y^{i+1}$ applied to $x$
  and $a$. If $P$ is a pole of $\Tr_{K_{N_*}/\Fbb_q(x)}(Q_if)$:

  \begin{align*}
    \nu_P(\Tr_{K_{N_*}/\Fbb_q(x)}(Q_if))\deg(P)
    &\geq \min_{\Pfrak|P}\left\lfloor\frac{\nu_{\Pfrak}(Q_i)+
    \nu_\Pfrak(f)}{e(\Pfrak|P)}\right\rfloor\deg(P)\\
    &\geq \sum_{\Pfrak|P}(\nu_\Pfrak(Q_i)+\nu_\Pfrak(f))\deg(\Pfrak)\\
    &\geq \sum_{\Pfrak|P}(\nu_\Pfrak(Q_i)-\nu_\Pfrak(A(D)))\deg(\Pfrak)\\
  \end{align*}
  It follows that
  \[\deg (\Tr_{K_{N_*}/\Fbb_q(x)}(Q_if))_- \leq \deg (Q_i)_-+\deg(A(D)).\]
  But $A(D)\leq D+\mathrm{Diff}(K_N)-2(x)_-+(a)_-$.
  Thus, since $\deg(\Diff(K_N)-2(x)_-)=2g-2$, where $g$ denotes the
  genus of $K_N$, $\deg A(D)\leq \deg(D)+d_x+2g-2$.
  Since $g\leq (d_x-1)(d_y-1)$ (see~\cite[Corollary~2.6]{Beelen09}), it
  follows that $\deg(A(D))=O(\deg(D)+d_xd_y)$ and 
  $\deg (\Tr_{K_{N_*}/\Fbb_q(x)}(Q_if))_-=O(d_xd_y+\deg(D))$.
  Thus, according to Corollary~\ref{coefficient_trace_formula},
  $\partial_yN_*(x,a)f$ has coefficients of degree $O(d_xd_y+\deg(D))$.
  Since $(\partial_yN_*(x,a)f)^{-1}$ has coefficients of size $O(d_xd_y)$, the
  result follows.
\end{proof}
\begin{Notation}\label{notation_T_P(B)}
  Let $\Bcal$ be a basis of $\Lcal(A(D))$ with $D\in\mathrm{Div}(K_{N})$, and $P\in\Fbb_q[x]$ 
  an irreducible polynomial verifying the hypothesis of
  Corollary~\ref{coeff_poles_trace}. Let $\Pfrak\in\Pbb(K_N)$ be lying over
  $P$ and $t_\Pfrak$ be a prime element of $\Pfrak$ and $B_0$ be an
  $\Fbb_p$-basis of $\Gcal_\Pfrak$.
  We denote by $\Tcal_\Pfrak(\Bcal)$ the matrix with coefficient in $\Fbb_p$
  whose columns are the Taylor expansion of the image of elements of
  $\Bcal$ by the map $f\mapsto f-S_{p-1}(f)$, at
  precision $\left\lfloor\frac{\deg A(D)}{\deg \Pfrak}\right\rfloor+1$ written in
  the basis $\Bcal_0\times(t_\pfrak^i)_{i\deg(\Pfrak)\leq\deg(A(D))}$.
\end{Notation}
We can now write the final version of our algorithm the solve the
$p$-Riccati equation in Algorithm~\ref{algo_p-riccati_irred}.\\
  \begin{algorithm}
  \begin{flushleft}
    \emph{Input:} $N_*\in \Fbb_{p^b}[x,y]$ an irreducible separable polynomial.\\
    \emph{Output:} $f\in K[a]$, where $a$ is a root of $N_*$ such that
    $f^{(p-1)}+f^{p}=a^p$, if such an $f$ exists. 

  \end{flushleft}
    \begin{enumerate}
      \item Test if $N_*^p(\partial)$ is irreducible using
        Algorithm~\ref{algo_irreducibility}.
      \item \textbf{If} $N_*^p(\partial)$ is irreducible \textbf{return}.
      \item Set $d_y:=\deg_Y N_*$ and
        $K_{N_*}:=K[a]=\nicefrac{\Fbb_{p^b}(x)[y]}{(N_*)}$.
      \item Compute $(a)_-$.
      \item Compute $A:=\mathrm{Diff}(K_{N_*})-2(x)_-$.
      \item Set $D:=0$, $n:=0$ and $l=0$.
      \item\label{step_random_div} Select
        $(D_1,\ldots,D_{n})\in\mathrm{Div}(K_{N})^{n}$ a family of
        $n$ randomly chosen effective divisors\\ of degrees $2g+1$
      \item $D\leftarrow D+D_1+\ldots D_n$, $l\leftarrow l+n$, $n\leftarrow
        n+\max(1,l)$ and $A(D):=A$.
      \item \textbf{For} $\Pfrak\in\Supp D$ \textbf{do:}
    \begin{itemize}
      \item $A(D)\leftarrow A(D)+\Pfrak$ 
    \end{itemize}
  \item $A(D)\leftarrow \max((a)_-,A(D))$.
  \item Compute a basis $\Bcal$ of $\Lcal(A(D))$
  \item Select $P\in\Fbb_q[x]$ an irreducible polynomial verifying the
    hypothesis of\\ Corollary~\ref{coeff_poles_trace} with respect to $D$ and $\Pfrak|P$.
  \item Compute the Taylor expansion  $V$ of $a$ in $t_\Pfrak$ at precision
    $\left\lfloor \frac{\deg A(D_1,\ldots,D_{g+1})}{\deg \Pfrak}\right\rfloor+1$
  \item Compute $\Tcal_{\Pfrak}(\Bcal)$ (see Notation~\ref{notation_T_P(B)}).
  \item Solve $\Tcal_{\Pfrak}(\Bcal)X=V$.
    \item \textbf{If} a solution $X$ exists reconstruct a solution to the
      $p$-Riccati equation from it\\
        and \textbf{return} it.
  \item \textbf{Else} redo from step~\ref{step_random_div}
  \end{enumerate}
    \caption{p-Riccati\_with\_irreducibility}
    \label{algo_p-riccati_irred}
  \end{algorithm}

\begin{Theoreme}\label{complexity_size_p-riccati}
  Let $r$ be the dimension of $\Gfrak^p_N$ over $\Fbb_p$, where
  $N(y^p)=N_*^p(y)$. We have $r\leq d_xd_y$ and Algorithm~\ref{algo_p-riccati_irred} returns 
  if it exists a solution of the $p$-Riccati equation relative to $N$ whose coefficients are of
  degree $O(rd_xd_y)$ at the cost of 
  \begin{itemize}
    \item testing the irreducibility of
      $N_*^p(\partial)$ using Algorithm~\ref{algo_irreducibility}
    \item factoring the divisors $(a)_-$ and $(x)_-$
    \item computing the different divisor of $K_{N}$
    \item selecting $O(r)$ uniformly random elements of $\mathrm{Div}(K_{N_*})$ of degree
      $2g+1$
    \item computing  $O(\log_2(r))$ basis of Riemann-Roch spaces of
      dimension $O(rd_xd_y)$.
    \item $\tilde{O}(bpr^2d_x^2d_y^3+(brd_xd_y)^{\omega})$ bit operations.
  \end{itemize}
  The total complexity of the computation is polynomial in $b$, $d_x$
  and $d_y$ and linear in $p$.
\end{Theoreme}
\begin{proof}
  The cost of steps $(1)$ to $(5)$ in Algorithm~\ref{algo_p-riccati_irred}
  is the cost of using Algorithm~\ref{algo_irreducibility}. The
  cost of step (4) is the cost of computing $(a)_-$, $(x)_-$ and 
  $\mathrm{Diff}(K_{N_*})$.\\
  The degree of $D$ roughly doubles at each repetition of steps $(7)$ to $(15)$ and are
  repeated on average $O(\log_2(r))$ times after which we have selected
  $O(r)$ uniformly random elements of $\Gfrak^p_N$ which form a generating
  family of it.
  The cost of steps
  (7) to (10) is essentially the cost of selecting uniformly random
  divisors of degree $2g+1$. By definition of $A(D)$, it
  is of degree $O(d_xd_y+rg)$. Since $g=O(d_xd_y)$ we find that
  $A(D)$ is of degree $O(rd_xd_y)$.\\
  Since we know that the solution of the $p$-Riccati equation constructed by
  Algorithm~\ref{algo_p-riccati_irred} is an element of
  $\Lcal(A(D))$, Proposition~\ref{divisor<->coeff} states
  that this solution has coefficients of degree $O(rd_xd_y)$.\\
  The cost of step (11) is thus the cost of computing a basis of
  $\Lcal(A(D))$ which is of dimension
  $O(\deg(A(D)) \subset O(rd_xd_y)$.\\
  Step (11) requires the computation of $\mathrm{Disc}(N_*)$ whose cost is
  negligible in regard of the final result.\\
  The cost of steps (13) and (14) is the cost of computing the Taylor
  expansions of $O(rd_xd_y)$ functions in $\Lcal(A(D))$ using Newton
  iterations. The cost for one such function is $\tilde{O}(bprd_xd_y^2)$ bit
  operations so the total cost is $\tilde{O}(bpr^2d_x^2d_y^3)$. Using this
  we can compute the Taylor expansions of $h-S_{p-1}(h)$ for $h\in\Bcal$
  at precision $O(\frac{rd_xd_y}{\deg(\Pfrak)})$. We recall that $\Bcal$ is
  an $\Fbb_{p^b}$-basis of $\Lcal(A(D))$. Since we want the result on an
  $\Fbb_p$-basis, we still need to multiply the result by an $\Fbb_p$-basis
  of $\Fbb_{p^b}$ which can be done in $\tilde{O}((brd_xd_y)^2\log(p))$ bit
  operations. Thus steps
  (13) and (14) can be done in $\tilde{O}(bpr^2d_x^2d_y^3+(brd_xd_y)^2)$ binary
  operations.\\
  Finally, step (14) is a matter of solving a $\Fbb_p$-linear system of
  size $O(brd_xd_y)\times O(brd_xd_y)$ which can be done in
  $\tilde{O}((b rd_xd_y)^{\omega})$ operations in $\Fbb_p$.\\
  Reconstructing the solution to the $p$-Riccati equation is a matter of
  summing $O(rd_xd_y^2)$ polynomial coefficients in $\Fbb_{p^b}[x]$ of degree
  $O(rd_xd_y)$ which can be done in $\tilde{O}(br^2d_x^2d_y^3\log(p))$ bit
  operations.

  The sum of those cost yield the final result.
\end{proof}

  \begin{Remarque}\label{iterative_improvement}
    In our experiments we often found that $r=O(1)$ hence the expression of
    the complexity in terms of this additional parameter and not purely in
    terms of $d_x$ and $d_y$.
  \end{Remarque}

\section{Factoring differential operators}\label{section_fact_eff}

Now that we have a working algorithm to solve $p$-Riccati equations and
degree bounds for the solutions, we discuss how it fits in the broader
context of differential operators factorisation. We begin by discussing how
to go from a solution of the $p$-Riccati equation relative to $N$, to the
corresponding irreducible divisor of $N(\partial^p)$.

\begin{Proposition}\label{Proposition_kernel_irred_div}
  Let $N\in C[Y]$ be a separable irreducible polynomial and $f\in K_N$ be a solution
  to the $p$-Riccati equation relative to $N$. If $L$ is a generator of the
  ideal of operators in $K\langle\partial\rangle$ which are left multiple
  of $\partial-f$ then $L$ is an irreducible divisor of $N(\partial^p)$.
\end{Proposition}
\begin{proof}
  We consider $K\langle\partial\rangle_{\leq\deg(N)}=\{L\in
  K\langle\partial\rangle|\mathrm{ord}(L)\leq\deg(N)\}$ and the $K$-linear
  map $\psi_N:K\langle\partial\rangle_{\leq\deg(N)}\rightarrow
  \nicefrac{K_N\langle\partial\rangle}{K_N\langle\partial\rangle(\partial-f)}$
  which maps an operator to its image modulo $\partial-f$. \\Since $\dim_K 
  K\langle\partial\rangle_{\leq\deg(N)}=\deg(N)+1$ and $\dim_K
  \nicefrac{K_N\langle\partial\rangle}{K_N\langle\partial\rangle(\partial-f)}=\deg(N)$,
  $\psi_N$ has a nontrivial kernel. In particular $\mathrm{ord}(L)\leq
  \deg(N)$. Let us show that $L$ is a divisor of $N(\partial^p)$. We claim
  that $\gcrd(L,N(\partial^p))$ is a multiple of $\partial-f$. Indeed,
  $\partial-f$ is a divisor of $\partial^p-y_N$ which is a divisor of
  $N(\partial^p)$ and is also a divisor of $L$. By definition of $L$,
  $\gcrd(L, N(\partial^p))=L$ and $L$ is a divisor of $N(\partial^p)$.
  Since $\mathrm{ord}(L)\leq\deg(N)$, it has to be irreducible according to
  Proposition~\ref{properties_sec_1}.(v).
\end{proof}
The proof of this result also points to an algorithmic way of deducing an
irreducible divisor of $N(\partial^p)$ from a solution to the $p$-Riccati
equation relative to $N$.
\begin{Corollaire}\label{better_way}
  Let $N\in C[Y]$ be an irreducible polynomial and $f\in K_N$ be a solution
  to the $p$-Riccati equation relative to $N$. Set $d_y=\deg(N)$.\\
  Let $a_0=1$ and for all $i\in\ldbrack0;d_y-1\rdbrack$,
  $a_{i+1}=a_if+a_i'$. Consider the matrix $M(f)$ in $M_{d_y,d_y+1}(K)$ whose columns
  are the coefficients of the $a_i$ (in some fixed $K$ basis of $K_N$).
  Then for any nonzero $v=(v_0,\ldots,v_{d_y})\in\ker(M)$,
  $\sum_{i=0}^{d_y}v_i\partial^i$ is an irreducible divisor of
  $N(\partial^p)$.
\end{Corollaire}
\begin{proof}
  We consider $K\langle\partial\rangle_{\leq\deg(N)}=\{L\in
  K\langle\partial\rangle|\mathrm{ord}(L)\leq\deg(N)\}$ and the $K$-linear
  map $\psi_N:K\langle\partial\rangle_{\leq\deg(N)}\rightarrow
  \nicefrac{K_N\langle\partial\rangle}{K_N\langle\partial\rangle(\partial-f)}$
  which maps an operator to its image modulo $\partial-f$. For dimensional
  reasons, we know that
  $\psi_N$ has a nontrivial kernel. Besides, any nonzero
  element of the kernel is a multiple of $\partial-f$ in
  $K\langle\partial\rangle$ of order less than $\deg(N)$. From
  Proposition~\ref{Proposition_kernel_irred_div} and
  Proposition~\ref{properties_sec_1}.(v) this means that it is a
  irreducible divisor of $N(\partial^p)$.
  We claim that the matrix $M(f)$ is the matrix of this restriction from the
  basis $(1,\partial,\ldots,\partial^d)$ to the $K$-basis of $K_N\simeq
  \text{\large$\nicefrac{K_N\langle\partial\rangle}{K_N\langle\partial\rangle(\partial-f)}$}$
  we have fixed.\\
  Indeed let $L'=\partial^k l_k+\partial^{k-1}l_{k-1}+\cdots+l_0$ be any
  differential operator in $K_N\langle\partial\rangle$. Then
  there exists an operator $B=\partial^{k-1}b_{k-1}+\cdots+\partial b_1+b_0\in
  K_N\langle\partial\rangle$ and $b_{-1}\in K_N$ such that
  \[L'=B(\partial-f)+b_{-1}.\]
  Then 
  \begin{align*}
    L'&=\sum_{i=0}^{k-1}\partial^{i+1}b_i-\sum_{i=0}^{k-1}\partial^i
    (b_i'+fbi)+b_{-1}\\
    &=\partial^kb_{k-1}+\sum_{i=0}^{k-1}\partial^i(b_{i-1}-b_i'-fb_i) 
  \end{align*}
    and we find that $l_i=b_{i-1}-b_{i}'-fb_{i}$, or equivalently 
  $b_{i-1}=l_i+b_{i}'+fb_{i}$ and $b_{k-1}=l_k$. We apply this result to
  $L'=\partial^k$. It immediately follows that the corresponding $b_{-1}$
  is the $k$-th term of the recursive sequence defined by $a_0=1$,
  $a_{i+1}=a_if+a_i'$, which concludes the proof.
\end{proof}
It is now easy to see that the coefficients of
$\gcrd(N(\partial^p),\varphi_N^{-1}(\partial-f))$ are of size independent
from $p$ as long as it is also the case for the coefficients of $f$, which we
know to hold true from Theorem~\ref{complexity_size_p-riccati}.

\begin{Lemme}
  We keep the notation of Corollary~\ref{better_way} with the additional
  hypothesis that $f\in\Lcal(A(D))$ where $D\in\mathrm{Div}(K_N)$ is a
  generating divisor of $\Gfrak^p_N$. Then for all $i\in\ldbrack
  1;d_y\rdbrack$,
  \[a_i\in\Lcal\big(iA(D)+(i-1)\max(\mathrm{Diff}(K_N)-2(x)_-,\,0)\big).\]
\end{Lemme}
\begin{proof}
  We know that $a_1=f\in\Lcal(A(D))$ so the proposition is verified here.
  We now suppose that the conclusion of the lemma holds for the index $i$.
  Let $\Pfrak\in\Pbb_{K_N}$ be a place and $t_\Pfrak$ be a prime element of
  it. Then 
  \[\nu_\Pfrak(a_i')\geq \nu_\Pfrak(a_i)+\nu_\Pfrak(t_\Pfrak')-1.\]
  In all generality, $1-\nu_\Pfrak(t_\Pfrak')$ is precisely one more than
  the valuation of $\mathrm{Diff}(K_N)-2(x)_-$ in $\Pfrak$. In particular
  if $\Pfrak$ is ramified then it is smaller than twice the valuation of
  $\mathrm{Diff}(K_N)-2(x)_-$ which is smaller than the valuation of $A(D)+\mathrm{Diff}(K_N)-
  2(x)_-$. If it is not ramified then either $a_i$ does not have a pole in
  $\Pfrak$, in which case neither does $a_i'$ or it has one and we have
  both $\nu_\Pfrak(a_i')\geq \nu_\Pfrak(a_i)-1$ and $\nu_\Pfrak(A(D))\geq 1$.
  Thus $\nu_\Pfrak(a_i)-\nu_\Pfrak(a_i')$ is once again smaller than the
  valuation of $A(D)+\mathrm{Diff}(K_N)-2(x)_-$.
  Therefore 
  $$a_i'\in\Lcal\big((i+1)A(D)+i\max(\mathrm{Diff}(K_N)-2(x)_-,0)\big.$$
  Furthermore
  since $f\in A(D)$, so too does $f a_i$ and the result follows.
\end{proof}
  \begin{algorithm}[H]
  \begin{flushleft}
    \emph{Input:} $N_*\in \Fbb_q(x)[Y]$ an irreducible separable polynomial, $f$ a
    solution of the $p$-Riccati equation relative to $N_*$.\\
    \emph{Output:} $L\in K\langle\partial\rangle$ the smallest monic multiple of
    $\partial-f$ with coefficients in $K$.
  \end{flushleft}
    \begin{enumerate}
      \item Set $K_{N_*}=\Fbb_q(x)[a]$ with $a$ a root of $N_*$.
      \item Set $d_y:=\deg N_*$.
      \item Set $a_0:=1$.
      \item \textbf{For} $i$ going from $1$ to $d_y$ \textbf{do}:
        \begin{itemize}
          \item Set $a_{i}:=a_{i-1}'+fa_{i-1}$
        \end{itemize}
      \item Set $M\in M_{d,d+1}(\Fbb_q(x))$ the matrix whose columns are
        the $a_i$\\written in the $\Fbb_q(x)$-basis $(1,a,\ldots,a^{d_y-1})$
        of $K_{N_*}$.
      \item Solve $MX=0$.
      \item Reconstruct $L$ from a solution and return it.
  \end{enumerate}
    \caption{Irreducible\_factors}
    \label{algo_p-riccati<->factor}
  \end{algorithm}
\begin{Theoreme}\label{size_p-riccati<->factor}
  Let $N_*\in\Fbb_{p^b}[x,y]$ be a separable irreducible polynomial. Keeping
  the notations of the previous sections, we suppose that
  $\dim_{\Fbb_p}\Gfrak^{p}_{N_*}=r$. Using
  Algorithm~\ref{algo_p-riccati_irred} we can compute a solution $f$ of the
  $p$-Riccati equation relative to $N$ whose coefficients are of degrees $O(rd_xd_y)$.
  Then Algorithm~\ref{algo_p-riccati<->factor} computes an irreducible
  divisor of $N_*^p(\partial)$ whose coefficients are of degree
  $O(rd_xd_y^3)$ in $\tilde{O}(rd_xd_y^{\omega+2})$ operations in
  $\Fbb_{p^b}$.
\end{Theoreme}
\begin{proof}
  The coefficients of the irreducible divisor returned by
  Algorithm~\ref{algo_p-riccati<->factor} can be expressed using the minors of the
  matrix $M$ whose columns are the $a_i$ written in the basis
  $(1,a,\ldots,a^{d_y-1})$. Since we know that $f$ has coefficients of
  degree $O(rd_xd_y)$, by immediate recurrence we get that $a_i$ has
  coefficients of degree $O(rd_xd_y^2)$. Thus the minors of $M$ are of
  degree $O(d_y^2rd_xd_y)$ since $M$ is a matrix of size $d\times(d+1)$.
  Furthermore, the coefficients $a_i$ can all be computed in
  $\tilde{O}(rd_xd_y^3)$ operations in $\Fbb_q$. It finally remains
  to solve a linear system of size $d\times (d+1)$ with coefficients in $\Fbb_q(x)$ of degree
  $O(rd_xd_y^2)$. This can be done in in $\tilde{O}(rd_xd_y^{\omega+2})$ operations in
  $\Fbb_{p^b}$~\cite{Storjohann03}.
\end{proof}
\begin{Remarque}
  In practice we have observed that the growth of the size of the
  coefficients, from those of the solution to the $p$-Riccati equation, to
  those of the corresponding irreducible divisor of $N(\partial^p)$, is
  only linear in $d_y$ (and not quadratic as shown in
  Theorem~\ref{size_p-riccati<->factor}). We infer that the situation is
  similar to seeking the minimal polynomial of an algebraic function in
  some $K[a]$.
\end{Remarque}
\newpage
\bibliographystyle{alpha}
\bibliography{../bibliographie.bib}

\newpage
\appendix
\section{Rational Puiseux Expansion for Global-to-Local transformations}
In this appendix we present a few facts on the link between places of an
algebraic function field and polynomial factorisation over local fields. We
also introduce the concept of rational Puiseux expansion developed
in~\cite{Duval89} by D. Duval which we use in
Algorithm~\ref{algo_irreducibility_rational} to compute parametrisations of
the completions $K_{N,\Pfrak}$ for the poles $\Pfrak$ of $x$ and $y_N$. We
use results from~\cite{PoWe21} to compute rational Puiseux expansions
efficiently up to the necessary precision.\\

Let $K$ be an algebraic function field and $L$ be a finite separable
extension of $K$ generated by some algebraic element $\alpha$ of minimal
polynomial over $K$ $\pi_\alpha\in K[T]$. It is a well-known fact that for any place $\Pfrak\in
\Pbb_K$, the set of places $\Pfrak'\in\Pbb_{L}$ lying above $\Pfrak$ is in
bijection with the irreducible factors of $\pi_\alpha$ in $K_\Pfrak[T]$. This can be
seen as a consequence of the following fact:

\begin{Lemme}\cite[§2~Theorem~2]{BourbakiTopVS}
  Let $F$ be a local field and $V$ be a Hausdorff topological $F$-vector space of
  dimension $n$. Then $V$ is isomorphic (as a topological vector space) to
  $F^n$.
\end{Lemme}
\begin{Corollaire}
  Let $F$ be a local field, $\nu$ be its valuation and $L/F$ be a finite
  extension of $F$. Then there is a unique valuation $\nu'$ on $L$ that
  extends $\nu$.
\end{Corollaire}
\begin{proof}
  Let $\Pfrak$ be the place of $F$ associated to $\nu$. The set of
  valuations extending $\nu$ is in bijection with the set of places above
  $\Pfrak$. Since
  $\sum_{\Pfrak'|\Pfrak}e_{\Pfrak'|\Pfrak}f_{\Pfrak'|\Pfrak}=[L:F]$ this
  set cannot be empty. If now $\nu_1$ and $\nu_2$ are two such valuations
  on $L$ then for $i\in\{1;2\}$, $\Pfrak_i$ is the set of $x\in L$ such
  that the sequence $x^n\xrightarrow[n\infty]{}0$ for the topology induced by
  $\nu_i$. But since $\nu_i$ makes $L$ a Hausdorff topological $F$-vector
  space of finite dimension, it is homeomorphic to $F^{[L:F]}$ which does
  not depend on $\nu_i$. Thus $\Pfrak_1=\Pfrak_2$ and $\nu_1=\nu_2$.
\end{proof}
\begin{Corollaire}\label{cor_place-factor}
  Let $K$, $L$, $\Pfrak$, $\alpha$ and $\pi_\alpha$ be defined as in the
  introduction. There is a bijection
  $\sigma:\{\Pfrak'\in\Pbb_L,\:\Pfrak'|\Pfrak\}\xrightarrow{\sim}\{\pi\in
  K_\Pfrak[T],\:\pi|\pi_\alpha\wedge\pi\text{ irreducible}\}$ which
  furthermore induces isomorphisms \[L_{\Pfrak'}\xrightarrow{\sim}
  \nicefrac{K_\Pfrak[T]}{\sigma(\Pfrak')}.\]
\end{Corollaire}
\begin{proof}
  Let $\Pfrak'|\Pfrak$. We denote by $\Ocal$ the valuation ring of $K$
  containing $\Pfrak$. There is a mapping $\iota_{\Pfrak'}:L\hookrightarrow
  L_{\Pfrak'}$. We can associate to $\Pfrak'$ the minimal
  polynomial of $\iota_{\Pfrak'}(\alpha)$ over $K_\Pfrak$ which we denote
  $\pi$ and is an irreducible factor of
  $\pi_\alpha$ in $K_\Pfrak[T]$. This mapping does not depend on the choice
  of $\iota_{\Pfrak'}$. Indeed if $\iota_1$ and $\iota_2$ are two such
  morphisms then $\iota_1\circ \iota_2^{-1}$ defines a morphism from
  $\iota_2(L)\rightarrow L_{\Pfrak'}$ which extends uniquely by continuity
  into an element of $\Gal{L_{\Pfrak'}}{K_\Pfrak}$. It follows that
  $\iota_1(\alpha)$ and $\iota_2(\alpha)$ are conjugated in $L_\Pfrak'$ and
  have the same minimal polynomial over $K_\Pfrak$.\\
  Note $T\mapsto \iota_{\Pfrak'}(\alpha)$ induces a monomorphism
  $\nicefrac{K_\Pfrak[T]}{\pi}\rightarrow L_{\Pfrak'}$ whose image contains
  $L$, so is dense in $L_\Pfrak'$, and is closed since the valuations of
  $\nicefrac{K_\Pfrak[T]}{\pi}$ and the restriction of that of
  $L_{\Pfrak'}$ must coincide. Thus it is also surjective which proves that
  $L_{\Pfrak'}$ and $\nicefrac{K_\Pfrak[T]}{\pi}$ are isomorphic.\\
  Conversely, if $\pi|\pi_\alpha$ is an
  irreducible factor of $\pi_\alpha$ in $K_\Pfrak[T]$ then it defines a
  separable extension $F_\pi/K_\Pfrak$ and a morphism
  $\varphi_\pi:L\hookrightarrow F_\pi$ which maps $K$ onto itself and $\alpha$ on a
  root of $\pi$ in $F_\pi$. Let $\nu$ be the unique valuation on $F_\pi$
  which extends that of $K_\Pfrak$ and $\Ocal_\pi=\{x\in F_\pi|\nu(x)\geq 0\}$.
  Then $\varphi_\pi^{-1}(\Ocal_\pi)$ is a valuation ring of $L$ containing
  $\Ocal$ so we can associate to it a unique place $\Pfrak'|\Pfrak$.
  Note that this place does not depend on the choice of the root of $\pi$
  that defines $\varphi_\pi$. Indeed, any other choice of a root of $\pi$
  defines a unique element of $\psi\in\Gal{F_\pi}{K_\Pfrak}$. Since
  $\nu\circ \psi$ is also a valuation of $F_\pi$ that extends the valuation
  of $K_\Pfrak$, $\nu\circ \psi=\nu$.\\
  The two maps thus defined are inverse of one another which
  follows immediately from the good definition of the two maps.
\end{proof}
We now restrain ourselves to the case $K=\Fbb_q(X)$ with $q=p^n$ and
$\Pfrak=X$. Let $P\in\Fbb_q[X,Y]$ be an irreducible polynomial and $L$ be
the rupture extension of $P$. In~\cite{Duval89}, Duval introduced the
notion of rational Puiseux expansion of a curve for fields of
characteristic $0$ but which can be extended to the case where $X$ is
tamely ramified in $L$. We recall
this notion here and show that these rational Puiseux expansions fully
describe a morphism $L\rightarrow L_\Pfrak$ for all the places above $X$.

\begin{Definition}\cite[Definition~2]{PoWe21}
  Let $P_1,\dots,P_\rho$ be the irreducible factors of $P$ in
  $\Fbb_q\ldbrack X\rdbrack[Y]$.
  A system of rational Puiseux expansions of $P$ over $\Fbb_q$ is a finite set
  $\{(\tilde{x_1},\tilde{y_1}),\dots,(\tilde{x_\rho},\tilde{y_\rho})\}\subset
  \overbar{\Fbb_p}\ldbrack T\rdbrack^2$ (where $\overbar{\Fbb_p}$ is an
  algebraic closure of $\Fbb_p$) such that
  \begin{enumerate}[label=\roman*)]
    \item $P_i(\tilde{x_k},\tilde{y_k})=0$ for all $k\in\ldbrack1;\rho\rdbrack$
    \item For all $k\in\ldbrack1;\rho\rdbrack$, $\tilde{x_k}=\gamma_k T^{e_k}$ for some
      $(\gamma_k,e_k)\in\overbar{\Fbb_p}^\times\times\Nbb^{*}$
    \item Each $e_k$ is ``minimal" which is to say that there is no
      $l\in\Nbb^*$ such that
      $(\tilde{x_k},\tilde{y_k})\in\overbar{\Fbb_p}\ldbrack T^l\rdbrack^2$.
  \end{enumerate}
\end{Definition}
\begin{Lemme}
  Let $\{(\tilde{x_1},\tilde{y_1}),\dots,(\tilde{x_\rho},\tilde{y_\rho})\}$
  be a system of rational Puiseux expansions of $P$ and $i\in\ldbrack
  1;\rho\rdbrack$. Let $\Kbb_i/\Fbb_q$ be the smallest extension of
  $\Fbb_q$ containing all the coefficients of $\tilde{x_i}$ and
  $\tilde{y_i}$. Let $\xi$ be a primitive $e_i$-th root of unity, where
  $\tilde{x_i}=\gamma_i T^{e_i}$ for some $\gamma_i$.
  For each $\sigma\in\Gal{\Kbb_i}{\Fbb_q}$ set $\sigma(\gamma_i)^{1/e_i}$ a
  $e_i$-th root of $\sigma(\gamma)$.
  The Puiseux series
  $S_{\sigma,j}=\sigma(\tilde{y_i})(\xi^{j}\sigma(\gamma_i)^{-1/e_i}X^{1/e_i})$
  are pairwise distincts roots of $P_i$.
\end{Lemme}
\begin{proof}
  In~\cite[Theorem~2]{Duval89}, Duval showed that the factorisation of $P$
  in $\Fbb_q((X))[Y]$ is given by $P=\prod P_i$ with
  $P_i=\prod_{\sigma,j}(Y-S_{i,j})$ which in particular proves that the
  $S_{i,j}$ are all distincts.  
\end{proof}
\begin{Proposition}\label{prop_ratpuiseux}
  Let $\{(\tilde{x_1},\tilde{y_1}),\dots,(\tilde{x_\rho},\tilde{y_\rho})\}$
  be a system of rational Puiseux expansions of $P$ and $i\in\ldbrack
  1;\rho\rdbrack$. Let $\Kbb_i/\Fbb_q$ be the smallest extension of
  $\Fbb_q$ containing all the coefficients of $\tilde{x_i}$ and
  $\tilde{y_i}$. Then the morphism
  \[\begin{array}{crcl}
    L&\hookrightarrow&\Kbb_i((T))\\
    X&\mapsto&\tilde{x_i}\\
  Y&\mapsto&\tilde{y_i}\end{array}\] is continuous for the topologie
  induced by the valuation associated to $\Pfrak_i$ and its image is dense
  in $\Kbb_i((T))$, which is to say that $\Kbb_i((T))$ is a representation
  of $L_{\Pfrak_i}$.
\end{Proposition}
\begin{proof}
  We begin by showing that $f_i=[\Kbb_i:\Fbb_q]<\infty$. Let
  $\Kbb_i'=\Kbb_i[\gamma_i^{1/e_i}]$ where $\tilde{x_i}=\gamma_i T^{e_i}$.
  Then $\tilde{y_i}(\gamma^{-1/e_i}T^{1/e_i)})$ is a root of $P$. Since
  those can only be in finite amount and
  $\Gal{\Kbb_i'}{\Fbb_q[\gamma^{1/e_i}]}$ maps
  Puiseux solutions of $P$ to other solutions of $P$, it follows that the
  orbit of $\tilde{y_i}(\gamma^{-1/e_i}T^{1/e_i)})$ must be finite. Taking
  $n$ to be the cardinal of its orbit it follows that each coefficient
  lives in $\Fbb_q[\gamma^{1/e_i}]^n$ and thus $\Kbb_i'$ is finite.\\
  Let $P_i$ be the irreducible factor of $P$ in $\Fbb_q((X))[Y]$ such that
  $P_i(\tilde{x_i},\tilde{y_i})=0$ and $\Pfrak_i$ be the corresponding
  place. $X\mapsto \tilde{x_i}$ and $Y\mapsto \tilde{y_i}$ induces a
  monomorphism $L_{\Pfrak_i}\simeq\nicefrac{\Fbb_q((X))[Y]}{P_i}\rightarrow
  \Kbb_i((T))$. It follows that the subfield of $L_{\Pfrak_i}$ of algebraic elements over
  $\Fbb_q$ is a subfield of $\Kbb_i$, thus $f_{\Pfrak_i}\leq f_i$.
  Furthermore, we know that there is a unique valuation on $\Kbb_i((T))$
  which extends that of $\Fbb_q((X))$ and so it must be $e_i^{-1}\nu_T$. As
  there is also a unique valuation on $L_{\Pfrak_i}$ that extends that of
  $\Fbb_q((X))$, it must be $e_i^{-1}\nu_{T|L_{\Pfrak_i}}$. Therefore, the
  ramification index of $\Pfrak_i$ is $e_i$.\\
  If $\xi$ is a primitive $e_i$-th root of the
  unity then the
  $S_{\sigma,j}=\sigma(\tilde{y_i})(\sigma(\gamma_i)^{-1/e_i}\xi^j X)$ for
  $(\sigma,j)\in\Gal{\Kbb_i}{\Fbb_q}\times\ldbrack1;e_i\rdbrack$
  are $e_if_i$ pairwise distinct roots of $P_i$. Therefore $\sum_{i=1}^\rho
  e_if_i\leq \deg(P)$. But we also know that $\sum_{i=1}^\rho e_i
  f_{\Pfrak_i}=\deg(P)$. Since $f_{\Pfrak_i}\leq f_i$ for all $i$ we can deduce
  that $f_i=f_{\Pfrak_i}$. Since now $\Kbb_i((T))$ and $L_{\Pfrak_i}$ have
  the same degree over $\Fbb_q((X))$ they must be equal.
\end{proof}

\begin{Theoreme}\label{theorem_compratPuiseux}\cite[Theorem~1]{PoWe21}~\\
  Let $P\in\Fbb_q[X,Y]$ of degree $d_Y$ in $Y$ and $\alpha$ be such that
  $P(X,\alpha)=0$. Let $\delta=\nu_X(\mathrm{Res}(P,\partial_YP))$. There
  exists an algorithm finishing in
  $\tilde{O}(d_Y\delta)$ arithmetic operations in $\Fbb_q$ and returning a set
  $\{(\tilde{x_1},\lceil y_1\rceil,\Kbb_1),\dots,(\tilde{x_l},\lceil y_l\rceil,\Kbb_l)\}$
  such that
  \begin{enumerate}[label=\roman*)]
    \item $(\tilde{x_i},\lceil y_i\rceil)\in\Kbb_i[T]\times\Kbb_i[T^{\pm1}]$ for all $i$.
    \item For each $i$, there exists a unique $\tilde{y_i}\in\Kbb_i((T))$
      such that $P(\tilde{x_i},\tilde{y_i})=0$ and $\tilde{y_i}=\lceil
      y_i\rceil+ O(T^{r_i+1})$ where $r_i$ is the maximal power of $T$ in
      $\lceil y_i\rceil$.
    \item The set
      $\{(\tilde{x_1},\tilde{y_1}),\dots,(\tilde{x_l},\tilde{y_l})\}$ is a system of rational Puiseux
      expansions of $P$.
    \item $\Kbb_i$ is the coefficient field of $(\tilde{x_i},\tilde{y_i})$.
  \end{enumerate}
  For no higher cost we may obtain the subset of non integral elements.
\end{Theoreme}
\begin{Remarque}
  To keep only the non integral elements of a system of rational Puiseux
  expansion, we can compute the whole system and keep only those of
  valuation negative. In truth the algorithm in~\cite{PoWe21} proceeds by
  computing a factorisation of $P$ in $K\ldbrack X\rdbrack[Y]$ as a product
  $P=uP_0P_\infty$ where $u\in \Fbb_q[X]$, $P_0$ is monic and $P_\infty$ is
  such that $P_\infty(0)\neq 0$ up to precision $\delta$. The algorithm
  then computes a system of rational Puiseux expansions of $P_0$ then a
  system of rational Puiseux expansions of
  $Y^{\deg_Y(P_\infty)}P_\infty(1/Y)$ up to a good enough precision before
  inverting them. For our purpose we can do only the latter part. Indeed if
  $(\tilde{x_i},\tilde{y_i})$ is a non integral rational Puiseux expansion
  then the associated Puiseux serie for $P$ has a non monic minimal polynomial
  $\pi\in\Fbb_p\ldbrack X\rdbrack[Y]$ with $\pi(0)\neq 0$. It follows that
  $\pi|P_\infty$.
\end{Remarque}
\end{document}